\documentclass{elsarticle}

\usepackage[utf8]{inputenc}
\usepackage{amsfonts} 
\usepackage{amsmath,amsthm,verbatim,amssymb,amsfonts,amscd, graphicx}
\usepackage{enumitem}
\usepackage[dvipsnames]{xcolor}
\usepackage{hyperref}
\usepackage[capitalise, noabbrev]{cleveref}
\usepackage{graphics}
\usepackage{caption}
\usepackage{subcaption}

\usepackage{lscape}

\newcommand{\R}[1][]{\mathbb{R}^{#1}}

\newcommand{\Pt}{P_t}
\newcommand{\Ps}{P_s}
\newcommand{\Pss}{P_{s'}}

\newcommand{\Qt}{Q_t}

\newcommand{\tab}[1]{\operatorname{table}\!\left[#1\right]}
\newcommand{\cost}[1]{w\left(#1\right)}
\newcommand{\cosst}{w}
\newcommand{\QQt}{\Qt}
\newcommand{\PPt}{\Pt}
\newcommand{\Z}{\mathbb{Z}}
\newcommand{\Zto}{\Z_2}

\newcommand{\im}{\textrm{Im}} 
\newcommand{\HasseGraph}{Hasse} 
\newcommand{\weight}{\omega}

\newcommand{\CmplxA}{K}

\newcommand{\ChnA}{U}
\newcommand{\ChnB}{U'}

\newcommand{\BndA}{W}
\newcommand{\BndB}{W'}
\newcommand{\BndC}{W''}
\newcommand{\symdif}{\triangle}

\newcommand{\Chnpic}{x}

\newtheorem{thm}{Theorem}[section]
\newtheorem{prop}[thm]{Proposition}
\newtheorem{lemma}[thm]{Lemma}
\theoremstyle{definition}
\newtheorem{definition}[thm]{Definition}
\newtheorem{remark}[thm]{Remark}

\begin{document}

\begin{frontmatter}

\title{The Parameterized Complexity of Finding Minimum Bounded Chains}

%% or include affiliations in footnotes:
\author[inf,cedas]{Nello Blaser} 
\author[math]{Morten Brun} 
\author[math]{Lars M. Salbu} 
\author[inf]{Erlend Raa Vågset} 
\ead{Erlend.Vagset@uib.no}

\address[inf]{Department of Informatics, University of Bergen, Bergen, Norway}
\address[math]{Department of Mathematics, University of Bergen, Bergen, Norway}
\address[cedas]{Center for Data Science, University of Bergen, Bergen, Norway}

%\maketitle
\begin{abstract}

Finding the smallest $d$-chain with a specific $(d-1)$-boundary in a simplicial complex is known as the \textsc{Minimum Bounded Chain} (MBC$_d$) problem. The MBC$_d$ problem is NP-hard for all $d\geq 2$. In this paper, we prove that it is also W[1]-hard for all $d\geq 2$, if we parameterize the problem by solution size. We also give an algorithm solving the MBC$_1$ problem in polynomial time and introduce and implemented two fixed parameter tractable (FPT) algorithms solving the MBC$_d$ problem for all $d$. 
The first algorithm is a generalized version of Dijkstra's algorithm and is parameterized by solution size and coface degree. The second algorithm is a dynamic programming approach based on treewidth, which has the same runtime as a lower bound we prove under the exponential time hypothesis.
\end{abstract}

\begin{keyword}
Computational Geometry\sep Topological Data Analysis\sep Algorithmic Topology\sep Minimum Bounded Chain\sep Parameterized Algorithms\sep Treewidth\sep Computational Complexity
\MSC[2020] 55N31\sep 62R40\sep 68W40
\end{keyword}

\end{frontmatter}

\section{Introduction} \label{sec:introduction}
% Introduce problem and related research
The \textsc{Minimum Bounded Chain} (MBC$_d$) problem in dimension $d$ is the problem of finding a minimum \(d\)-chain \(\BndA\), whose boundary \(\partial \BndA\) is a given \((d-1)\)-cycle \(\ChnA\). This problem is a useful generalization of the shortest path problem. It has been applied to 3d image segmentation \cite{Grady2010} and to find representative cycles in persistent homology \cite{escolar,emmett2015multiscale} (see \cref{Fig: motiverende eksempel}).

\begin{figure}[!h]
\centering
\includegraphics[width=\textwidth]{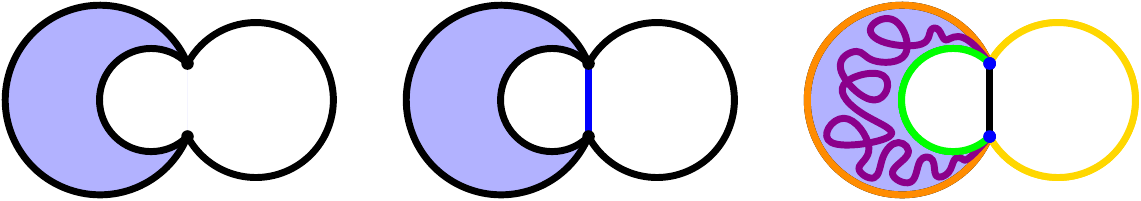}
\caption{\label{Fig: motiverende eksempel} A space at filtration step $t-1$ (left) where the addition of the blue edge at time step $t$ in the filtration (middle) gives birth to a persistent cycle. In the rightmost figure we have highlighted some options for representatives of this cycle: the dark magenta, orange, yellow, or green path together with the new edge. Several measures have been used for deciding which cycle is the best representative. One way is to give each edge a weight and to choose the cycle whose edges sum up to the lowest total weight. If we assume these weights to be proportional to length in this figure, we get that the green path together with the blue edge is the cycle with lowest total weight. Solving the MBC$_1$ problem at filtration step $t-1$ with the boundary of the blue edge as boundary gives this green path.}
\end{figure}

The problem with coefficients in \(\Z\) and \(\R\) has previously received some attention \cite{Sullivan1990,Dunfield2011,Chambers2015}. Here we study the problem with coefficients in \(\Zto\). The MBC$_d$ problem with coefficients in \(\Zto\) is known to be NP-hard to approximate \cite{borradaile_et_al:LIPIcs:2020:12179}. Previous study of the MBC$_d$ problem with coefficients in \(\Zto\) focus either on the case of \(d\)-dimensional simplicial complexes embeddable in \(\R[d+1]\) \cite{borradaile_et_al:LIPIcs:2020:12179}, or specific types of input cycles \cite{chen2011hardness,erlend_homloc}. Here we treat the general case of the MBC$_d$ problem with coefficients in \(\Zto\). This can be viewed as a special case of the algebraic problem known as \textsc{maximum likelihood decoding} (MLD) \cite{mld}. 

A recent paper shows that the MBC$_d$ problem with coefficients in \(\Zto\) is fixed-parameter tractable with respect to
the treewidth of the 1-skeleton \cite[Theorem 1.2]{black2021finding}. However, they also show that the treewidth of the $d$'th level of the Hasse diagram is bounded by the treewidth of the 1-skeleton.

% Our contributions
We approach the problem from the point of view of parameterized algorithms. We consider three different parameters, namely 
\begin{enumerate}[label=(\roman*)]
\itemsep0pt
    \item the maximum number of cofaces of codimension 1 of any $(d-1)$-simplex in the simplicial complex, termed \emph{coface degree} \(c\);
    \item the solution size \(k\), i.e. the number of simplices in an optimal solution;
    \item the treewidth \(\tau\) of the d'th level of the Hasse diagram of the simplicial complex. 
\end{enumerate}
In addition we write $n$ for the number of $(d-1)$-simplices, and note that the number of $d$-simplices is given as a polynomial of $n$.

Our contributions are as follows: 

\begin{thm}[Sec. \ref{sec: mbc1-section}]\label{thm:mbc1-in-P}
The MBC$_1$ problem is in P.
\end{thm}

\begin{thm}[Sec. \ref{sec:algorithm}] \label{thm:solution-coface-algorithm}
The MBC$_d$ problem can be solved in $c^k\operatorname{poly}(n)$-time. 
\end{thm}

\begin{thm}[Sec. \ref{sec: cofacesection}] \label{thm:coface-hardness}
The MBC$_d$ problem is NP-complete for $d\geq 2$, even when restricted to spaces with coface degree three. 
\end{thm}

\begin{thm}[Sec. \ref{sec:solutionsize}] \label{thm:solution-hardness}
The MBC$_d$ problem is W[1]-hard when parameterized by solution size. 
\end{thm}

\begin{thm}[Sec. \ref{sec:doubleparameter}]  \label{thm:solution-coface-hardness}
Unless the exponential time hypothesis is false, the MBC$_d$ problem can not be solved in $2^{o(\sqrt{k}\log(c))}\operatorname{poly}(n)$-time for any $d\geq2$. 
\end{thm}

\begin{thm}[Sec. \ref{sec: fpt-algorithm}] \label{thm:treewidth-algorithm}
The MBC$_d$ problem can be solved in $\mathcal{O}(2^{2\tau}\tau^2n)$-time. 
\end{thm}

\begin{thm}[Sec. \ref{sec:eth-tightness}]  \label{thm:treewidth-hardness}
The MBC$_d$ problem can not be solved in $2^{o(\tau)}\operatorname{poly}(n)$-time for any $d\geq2$, unless the exponential time hypothesis is false. 
\end{thm}
For \cref{thm:solution-coface-algorithm} and \cref{thm:treewidth-algorithm} we give explicit algorithms that solve the problem in the stated times, with implementations available at \url{https://github.com/lar-sal/PersHomLoc}. 

% Outline
This paper is structured as follows. In \cref{sec:preliminaries}, we formally define the concepts and problems used throughout this paper. In \cref{sec: mbc1-section} we show that the MBC$_1$ problem is solvable in polynomial time. Next, in \cref{sec:dijkstra} we give our solution to the MBC$_d$ problem using an algorithm inspired by Dijkstra's shortest path algorithm. We also prove hardness results relevant to this algorithm. In \cref{sec:treewidth}, we present the treewidth based algorithm and outline the proof showing that this algorithm is ETH-optimal. Finally, in \cref{sec:conclusions}, we reflect on our results and give some open problems for future research. 

\section{Preliminaries} \label{sec:preliminaries}

This section introduce notation, concepts and definitions that are used frequently throughout the paper. 

\subsection{Minimum Bounded Chains}\label{SEC: miminum bounded chain prelim}

Let $S$ be a set and let $\Zto$ be the field on two elements. The vector space generated by $S$ with coefficients in $\Zto$ is denoted by $\Zto[S]$. Elements in $\Zto[S]$ can be thought of as subsets of $S$, using the bijection mapping the vector $\sum_{s\in S}a_s s$ in $\Zto[S]$ to the subset $\{s\in S\,|\,a_s=1\}$ of $S$. This bijection is an isomorphism when viewing $\mathcal{P}(S)$ (the powerset of $S$) as a vector space where we let \emph{symmetric difference} $X\symdif Y:=(X\cup Y)\setminus (X\cap Y)$ act as addition and define scalar multiplication as $0\cdot X = \emptyset$ and $1\cdot X = X$ for any set $X \subseteq S$.

A \emph{(finite) simplicial complex} $K$ is a (finite) family of sets (called \emph{simplices}) closed under inclusion, i.e. if $\sigma\in K$ and $\rho \subseteq \sigma$ then $\rho \in K$. We say that $\rho$ is a \emph{face} of $\sigma$ and that $\sigma$ is a \emph{coface} of $\rho$. A \emph{$d$-dimensional simplex} (or a \emph{$d$-simplex}) is a simplex containing $d+1$ elements. A coface $\rho \subseteq \sigma$ has codimension $i$ if $\text{dim}(\sigma)= \text{dim}(\rho)+i$. The set of $d$-simplices in $K$ is denoted by $K^d$. We often draw geometric representations of simplicial complexes (see \cref{Fig: simplicial complex figure}).

\begin{figure}[!h]
\centering
\includegraphics[width=0.7\textwidth]{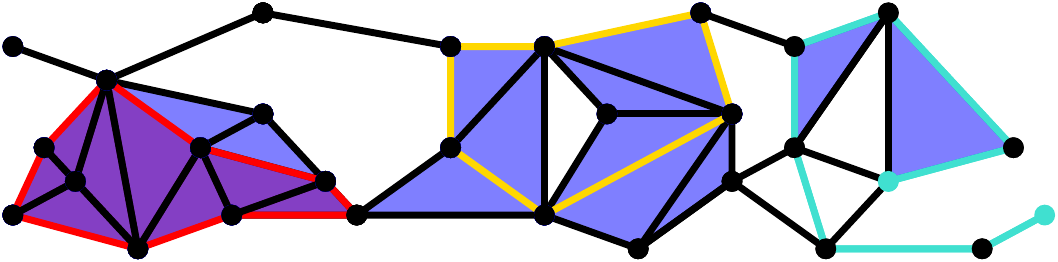}
\caption{\label{Fig: simplicial complex figure}A geometric representation of a $2$-dimensional simplicial complex. Within it we have highlighted a $2$-dimensional chain (purple) and its $1$-dimensional boundary (red), a $1$-dimensional cycle that is not a boundary of a $2$-dimensional chain (yellow) and a $1$-dimensional chain that is not a cycle (turquoise).}
\end{figure}

A vector in $C_d(K):=\Zto[K^d]$ is called a \emph{$d$-chain} in $K$. The \emph{boundary $\partial \sigma$ of a $d$-simplex} $\sigma$, is the $(d-1)$-chain that is the sum of the $(d-1)$-faces of $\sigma$. The \emph{boundary of the $d$-chain} $W$ is the $(d-1)$-chain $\partial W$, defined as the sum $\partial W = \sum_{\sigma\,|\,a_\sigma =1} \partial\sigma$. Alternatively, $\partial W$ contains a $(d-1)$-simplex $\rho$ if and only if $\rho$ is the codimension $1$ face of an odd number of simplices in $W$. A $d$-chain with an empty boundary is called a \emph{$d$-cycle}, and the subgroup of $d$-cycles is denoted by $Z_d(K):=\ker \,\partial\subseteq C_d(K)$. Meanwhile, the $d$-chains that are boundaries of some $(d+1)$-chain are called \emph{$d$-boundaries} and they form a subgroup denoted by $B_d(K):=\im \,\partial\subseteq Z_d(K) \subseteq C_d(K)$. In \cref{Fig: simplicial complex figure} we have some examples of different chains, cycles and boundaries.

A \emph{weighted simplicial complex} $(K,\weight)$ is a simplicial complex $K$ together with a family of positive real numbers $\weight=\{w_\sigma | {\sigma\in K}\}$. The number $w_\sigma>0$ is the \emph{weight} of the simplex $\sigma$, and the \emph{weight of a $d$-chain} is the sum of the weights of the simplices it contains. For unweighted simplicial complexes, we assign weight \(1\) to all simplices, such that the weight of a chain \(\BndA\) corresponds to the number of simplices in unweighted simplicial complexes. %, i.e. the weight of the $d$-chain $c\in C_d(K)$ is given by $\weight(c)=\sum_{\rho\,|\,c_\rho =1}w_\rho$. 

This paper is primarily about the problem of finding the smallest chain whose image under $\partial$ is some given boundary. More formally we look at the following problem:

\begin{definition} \label{MBC}
(Weighted) \textsc{Minimum Bounded Chain} (MBC$_d$) problem\\
INPUT: A (weighted) simplicial complex $\CmplxA$ and a $(d-1)$-boundary $\ChnA$.\\
OUTPUT: A $d$-chain $\BndA$ in $C_{d}(K)$ such that $\ChnA = \partial \BndA$.\\
MINIMIZING: The weight of $\BndA$.
\end{definition}

An example of the MBC$_1$ and the MBC$_2$ problems are given in \cref{FIG: mbc_example}. There are also more restrictive versions of this problem. In this paper, we are also interested in the problem where we know that boundary we are given is particularly simple. A \emph{spherical chain} is a $d$-cycle whose closure (when viewed as a subspace of the simplicial complex) is homotopy equivalent to a $d$-sphere.

\begin{definition} \label{SMBC}
\textsc{Spherical Minimum Bounded Chain} (SMBC$_d$) problem\\
INPUT: A simplicial complex $\CmplxA$ and a spherical $(d-1)$-boundary. $\ChnA$\\
OUTPUT: A $d$-chain $\BndA$ in $C_{d}(K)$ (if it exists) such that $\ChnA = \partial \BndA$.\\
MINIMIZING: The weight of $\BndA$.
\end{definition}

\begin{figure}[!h]
\centering
\includegraphics[width=\textwidth]{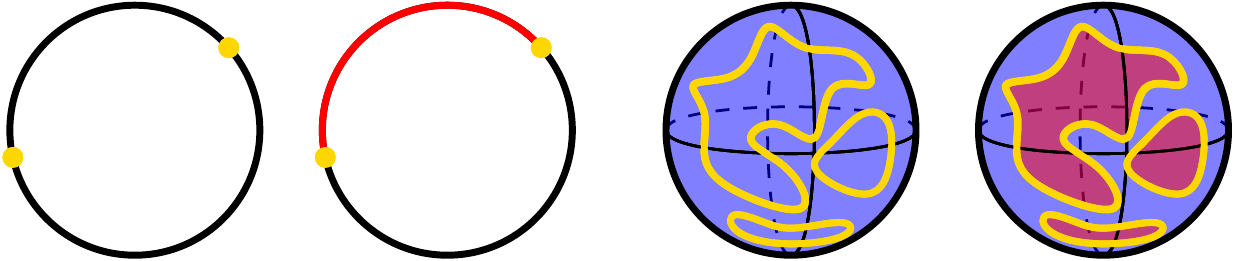}
\caption{\label{FIG: mbc_example} Two examples of minimum bounded chains, where the input boundary is in yellow, and the bounding chain is red. The first example is an instance of both the MBC$_1$ problem and the SMBC$_1$ problem. The second example is only an instance of the MBC$_2$ problem, since the yellow boundary is not a $1$-sphere (i.e. a circle) but the disjoint union of three circles. 
\label{FIG:smooth2d}}
\end{figure}

Our main reason for defining the SMBC$_d$ problem is that by proving this problem to be computationally hard we also prove that the problem of locating persistent cycles is hard. This is because finding the location of a persistent cycle is the same as finding the smallest chain that has the same boundary as the simplex added when the persistent cycle is born (see \cref{Fig: motiverende eksempel}).

By looking at the boundary matrix, the MBC$_d$ problem reduces to a special case of a well-known algebraic problem. 
\begin{definition}\label{MLD}
The \textsc{Maximum Likelihood Decoding} (MLD) problem\\
INPUT: A matrix $A \in \Zto^{m\times n}$, a target vector $u \in \Zto^m$ and a non-negative weight vector \(\weight \in \R[n]\).\\
OUTPUT: A vector \(x \in \Zto^n\) such that \(Ax=u\).\\
MINIMIZING: The weight of the vector $\cosst(x)=\weight^Tx$.
\end{definition}
The algorithms presented in this paper actually solve the maximum likelihood decoding problem, although the considered parameters have a more natural interpretation for the MBC$_d$ problem.

\subsection{Computational Complexity}
Computational complexity is all about how fast an algorithm can solve  a given problem (i.e. the \emph{runtime} of the algorithm). The gold standard are the algorithms that have a runtime that is polynomial in the input-size. The family of problems that can be solved in polynomial time is referred to as P. 

A problem is said to be in NP if the correctness of each solution to the problem can be verified in polynomial time and if a brute-force search algorithm can actually find a solution by trying all of them. A problem is said to be NP-complete if it is in NP and if solving the problem in polynomial time means that we can solve every problem in NP in polynomial time.

Many of the problems in this paper (and in life generally) are known to be \emph{NP-complete}. If the famous P$\neq$NP conjecture is true, then no NP-complete problem can be solved in polynomial time. Despite this, we want to solve these problems ``in practice'', and there are roughly five main (non-disjoint) frameworks that are used to describe what ``solving something in practice'' means:

%Computational complexity talks about how quick an algorithm can solve (i.e. the \emph{runtime}) a given computational problem. The gold standard are algorithms that have a runtime given as a polynomial of the input-size. Such polynomial time algorithms are called \emph{efficient}, and problems that can be solved by such algorithms are said to be \emph{in P}. There is a specific class of problems, called \emph{NP-complete} problems, for which no efficient algorithm has been found, and he famous conjecture P$\neq$NP that states that no such algorithms exist. Despite being difficult/slow to solve, many NP-complete problems are of great practical interest.

%There are roughly five main frameworks developed to answer these kind of questions:
\begin{enumerate}
\itemsep0pt
    \item Approximation: Develop algorithms to find near-optimal solutions.
    \item Randomization: Use randomization to get a better expected runtime.
    \item Restriction: Restrict attention to special cases.
    \item Heuristics: Find algorithms that often work well but that we can't prove are always fast and correct at the same time.
    \item Parameterization: Design algorithms that are polynomial whenever parameters describing the problem are fixed to a constant value.
\end{enumerate}

We focus on parameterized algorithms in this paper. \cref{sec:parameterized algorithms} and \cref{sec:W1-hardness and ETH} give a brief introduction to the concepts we need from parameterized complexity theory. For a more in-depth understanding of this field, consult either of the textbooks \cite{DowneyFellows2013, PA}.

\subsubsection{Parameterized Algorithms}\label{sec:parameterized algorithms}

A \emph{parameter} is a number associated to each instance/input of a computational problem. The parameters typically describe some property of the problem instance (e.g. coface degree or treewidth) or its solution (e.g. solution size). A computational problem together with specified parameters is called a \emph{parameterized problem}. A \emph{parameterized algorithm} is an algorithm solving a parameterized problem where we measure the runtime as a function of both the parameter and the input size. This allows for a more fine-grained analysis of the computational complexity of hard computational problems where the goal is to find parameterized algorithms that are provably efficient whenever the parameter is small.

We want parameterized algorithms that are \emph{fixed parameter tractable} (FPT), meaning that the expected exponential explosion in runtime is confined to the parameter alone. More precisely, if $k$ is the parameter, $n$ is the input size and $f$ is a computable functions then an FPT-algorithm has a runtime on the form $f(k)\cdot\operatorname{poly}(n)$. A problem is said to be in \textbf{FPT} if it can be solved by some \emph{FPT-algorithm}. 

Problems in \textbf{FPT} are often contrasted with those in \textbf{XP}, that are only solvable by some \emph{XP-algorithm}. These are algorithm with runtime of the form $\mathcal{O}(n^{g(k)})$ where $g$ is some computable function. \cref{table: XP vs FPT table} shows how the archetypical FPT-runtime $2^kn$ compares with the archetypical XP-runtime $n^{k+1}$ for various values of $k$ and $n$. In addition to this, \textbf{FPT} is known to be a strict subset of \textbf{XP} \cite[Proposition 27.1.1]{DowneyFellows2013}.
\begin{table}[h]
\begin{center}
\begin{tabular}{|l||*{5}{c|}}\hline
&\makebox[4em]{$n=50$}&\makebox[4em]{$n=100$}&\makebox[4em]{$n=150$}\\\hline\hline
$k=2$ & 625 & 2,500 & 5,625\\\hline
$k=3$ &15,625&125,000& 421,875\\\hline
$k=5$ &390,625&6,250,000 & 31,640,625\\\hline
$k=10$ &$1.9 \times 10^{12}$&$9.8 \times 10^{14}$&$3.7 \times 10^{16}$\\\hline
$k=20$ &$1.8 \times 10^{26}$&$9.5 \times 10^{31}$&$2.1 \times 10^{35}$\\\hline
\end{tabular}
\caption{\label{table: XP vs FPT table} The ratio $\frac{n^{k+1}}{2^kn}$ for various values of $n$ and $k$.}
\end{center}
\end{table}

There are often trivial brute force XP-algorithms that solve NP-complete problems. For example, the MBC$_d$ problem parameterized by solution size $k$ has a trivial $n^{\mathcal{O}(k)}$-time algorithm. Simply compute the boundary of every $d$-chain $|\BndA| \leq k$ and keep track of the smallest $W$ you find where $\partial\BndA = \ChnA$. Coming up with FPT-algorithms is often much harder as they typically require some level of insight into the computational problem in their design. 

\subsubsection{W[1]-Hardness and the ETH}\label{sec:W1-hardness and ETH}

Some parameterized problems do not appear to be solvable by any FPT-algorithm at all. One way of resolving these situation is to show that the problem is \emph{ParaNP-complete}, i.e. that it is NP-complete even when the parameter in question is constant. In this case, the problem can not be solved by either an FPT- or an XP-algorithm (assuming P$\neq$NP). Unfortunately, this means that we can't use this particular trick to prove that a problem has no FPT-algorithm if it has an XP-algorithm. To tackle this, we need hardness hypothesises that have been specially tailored to parameterized algorithms.

In particular, we are interested in two hardness hypothesises commonly used in parameterized complexity theory. The first is \textbf{FPT} $\neq$ W[1] (see \cite[Chapter 21]{DowneyFellows2013} for the exact definition of W[1]), which is a hypothesis similar to P$\neq$NP. We say that a parameterized problem is $W[1]$-hard if solving it in FPT-times implies \textbf{FPT}$=$W[1]. We can prove that a parameterized problem B is $W[1]$-hard we show that there is a \emph{parameterized reduction} from A to B, where A is parameterized problem already known to be $W[1]$-hard.

\begin{definition}[Parameterized reduction]
Let A and B be parameterized problems. A \emph{parameterized reduction} $F\colon A \to B$ is a function mapping instances $(X,k)$ of $A$ to instances $(Y,l)$ of $B$ in such a way that
\begin{itemize}
    \itemsep0pt
    \item $F(X,k)$ can be computed in FPT-time.
    \item $l \leq g(k)$ for some computable function $g$.
    \item $(X,k)$ is a ``yes'' instance if and only if $(Y,l)$ is a ``yes'' instance.
\end{itemize}
\end{definition}

The second hypothesis is the exponential time hypothesis (ETH). It is primarily based on the fact that no one has been able to solve $3$-SAT\footnote{$3$-SAT is the computational problem of determining the satisfiability of a formula in conjunctive normal form where each clause contains at most three literals.} in sub-exponential time. If we let $n$ be the number of variables in a $3$-SAT formula, then we have:

\begin{definition}[ETH] $3$-SAT cannot be solved in $2^{o(n)}$-time.
\end{definition}

While the ETH is perhaps the easiest of the hypothesis to understand, it is also the strongest assumption made in this paper. Explicitly, the ETH implies \textbf{FPT}$\neq$W[1] which in turn implies P$\neq$NP. See the references mentioned at the start of this section for further details.

\section{The \texorpdfstring{MBC$_1$}{MBC\_1} problem is in P} \label{sec: mbc1-section}

In order to solve the MBC$_d$ problem we only need information about the $d$ and $d-1$ dimensional simplices. We can therefore restrict attention to the simplicial complexes that are $1$-dimensional (i.e. graphs) when $d=1$, without loss of generality. If the simplicial complex given as input has a higher dimension, we can just solve the problem on its $1$-skeleton.

We now formulate the MBC$_1$ problem using graph theoretical language. Let the input be a finite, simple and undirected graph denoted by $G$ and let $V(G)$ denote the set of vertices and $E(G)$ be the set of edges in the graph. A $0$-chain $\ChnA$ is a subset of vertices while a $1$-chain $\BndA$ is a subset of edges. We write $\partial\BndA=\ChnA$ if and only if the set of vertices having odd degree in the subgraph $G'= (V(G), \BndA)\subseteq (V(G),E(G))$ is equal to $\ChnA$. Using the handshaking lemma (i.e. $\sum_{v\in V(G')}\deg_{G'}(v)=2 |E(G')|$), we know that the $0$-chain $\ChnA$ must contain an even number of vertices if it is the boundary of a $1$-chain.

A \emph{walk} in a graph from $v$ to $u$ is a sequence of edges $(e_1,e_2,\dots e_r) = ((v,x_1),(x_1,x_2),\dots,(x_{r-1},u)$, and it is said to be \emph{cyclic} if $v=u$. A \emph{trail} in a graph is a walk where every edge is traversed no more than once. We say that a $1$-chain $\BndA$ is \emph{acyclic} if the subgraph $G' = (V(G),\BndA)$ does not contain a cyclic trail. Note that a cyclic trail visits any vertex in the graph an even number of times. This means that removing the edges of such a trail contained in a $1$-chain reduces the weight of that chain without changing its boundary (pictured in \cref{FIG: Part1 MBC1 section}). This proves the following lemma.

\begin{lemma}
    Any optimal solution to the MBC$_1$ problem is an acyclic $1$-chain. \qed
\end{lemma}

\begin{figure}[!h]
\centering
\includegraphics[width = \textwidth]{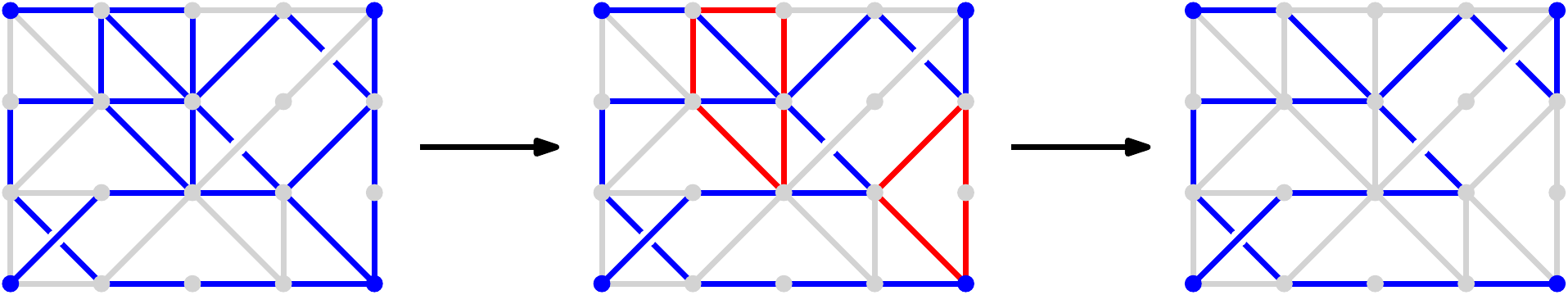}
\caption{\label{FIG: Part1 MBC1 section} A $1$-chain (left) containing two cyclic trails marked in red (middle) and an acyclic $1$-chain with the same boundary obtained by removing the red cyclic trails (right).}
\end{figure}

Let $\ChnA$ be a $0$-chain where $|\ChnA|=2n$. We define a \emph{pairing of $\ChnA$} to be a collection of $n$ pairs $(v_i,u_{i})_{i=0}^n$ covering $\ChnA$, i.e. $\bigcup_{i=0}^n \{v_i,u_i\}=\ChnA$. We denote the set of edges on a trail from $v$ to $u$ as $t(v,u)$.

\begin{lemma}\label{mbctraillemma}
Let $\ChnA$ be a $0$-chain with $|\ChnA|=2n$. If $\BndA$ is an acyclic $1$-chain with $\partial\BndA = \ChnA$, then there exist a pairing $(v_i,u_{i})_{i=0}^n$ of $\ChnA$ and edge-disjoint trails $t(v_i,u_i)$ so that 
\begin{equation*}
    \BndA = \bigsqcup_{i=0}^n t(v_i,u_i).
\end{equation*}

\end{lemma}

\begin{proof}
We use induction on $n$. 

The base case $n=0$ holds trivially, so assume that the hypothesis is true for $n-1$ and consider the case $|\ChnA|=2n$. Let $W$ be an acyclic chain with $\partial\BndA = \ChnA$, where $\BndA=E(G')$ for some subgraph $G'\subseteq G$. Let $v$ be a vertex in $\ChnA$. Then there is a trail $t(v,u)$ in $G'$ between $v$ and some other vertex $u$ in $\ChnA$. %This trail can be found algorithmically by adding arbitrary edges from $G'$ that we have not traversed to a trail starting at $u$. When we meet a vertex of even degree, we know that we can leave the vertex. Once we meet a vertex $u$ with odd degree in $G'$ we end our trail. 

Now, $\BndB = \BndA\setminus t(v,u)$ is itself going to be an acyclic chain with boundary $\ChnB= \ChnA\setminus\{v,u\}$, where $|\ChnB|=2(n-1)$. By the induction hypothesis there is a pairing $(v_i,u_{i})_{i=0}^{n-1}$ of $\ChnB$ and edge-disjoint trails $t(v_i,u_i)$ so that $\BndB = \bigsqcup_{i=0}^n t(v_i,u_i)$. Adding $(v,u)$ to the pairing of $\ChnB$ gives a pairing of $\ChnA$, and adding $t(v,u)$ to the collection of trails completes the proof of the inductive step.
\end{proof}

For the acyclic $1$-cycle in \cref{FIG: Part1 MBC1 section}, we have a pairing with edge-disjoint trails colored in \cref{FIG: Part2 MBC1 section}. Next, let the shortest path from $v$ to $u$ in some graph $G$ be denoted as $p(v,u)$.

\begin{prop}\label{mbc1disjointpaths}
Optimal solutions $\BndA$ to the MBC$_1$ problem on input $\ChnA$ are disjoint unions of shortest paths between pairs $(v_i,u_{i})_{i=0}^n$ in a pairing of $\ChnA$, i.e.
\begin{equation}
    \BndA = \bigsqcup_{i=0}^n p(v_i,u_i).
\end{equation}

\end{prop}
\begin{proof}
We give a proof by contradiction, where the idea is pictured in \cref{FIG: Part2 MBC1 section}. First, let $\BndA$ be an optimal solution of MBC$_1$. By \cref{mbctraillemma} we know that there exist a pairing $(v_i,u_{i})_{i=0}^n$ of $\ChnA$ and edge-disjoint trails $t(v_i,u_i)$ so that $\BndA = \bigsqcup_{i=0}^n t(v_i,u_i)$. 

In order to get a contradiction, we assume that at least one of these trails, say $t(v_i,u_i)$, is not a shortest path between $v_i$ and $u_i$ in $G$. 

Let $p(v_i,u_i)$ denote an actual shortest path. Then we know that the $1$-chain given by the symmetric difference 
\begin{equation*}
    \BndB = (\BndA \setminus t(v_i,u_i))\,\symdif\, p(v_i,u_i)
\end{equation*}

have the same boundary as $\BndA$. This is because taking the symmetric difference of a trail changes precisely the odd-even parity of its start and end vertex (and nothing else). The chain $\BndB$ clearly has a weight lower than $\BndA$, since

\begin{align*}
    \cost{\BndB} &= \cost{\BndA \setminus t(v_i,u_i))\,\symdif\, p(v_i,u_i)} \\
    &\leq \cost{ \BndA}- \cost{t(v_i,u_i)} + \cost{p(v_i,u_i)} \\
    &<   \cost{ \BndA}- \cost{t(v_i,u_i)} + \cost{t(v_i,u_i)} =\cost{ \BndA}.
\end{align*}
This contradicts the minimality of \(W\).
\end{proof}

\begin{figure}[!h]
\centering
\includegraphics[width = \textwidth]{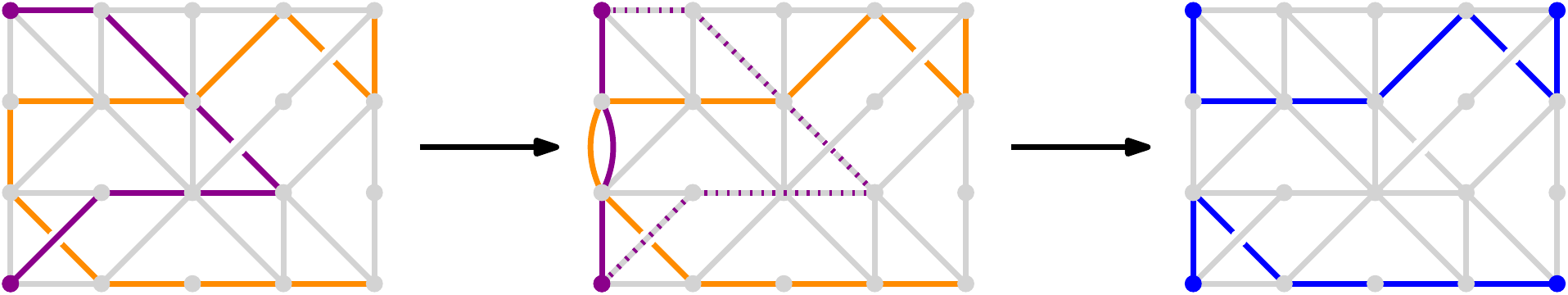}
\caption{\label{FIG: Part2 MBC1 section}A partitioning of the (rightmost) $1$-chain from \cref{FIG: Part1 MBC1 section} (left), removing the magenta path and replacing it with the shortest path between its endpoints to the chain (middle) gives an even smaller $1$-chain (right). Note that if the same edge appears both an old path and in the new path then these cancel out because we ``added'' the path using symmetric difference.}
\end{figure}

Note that when the boundary $\ChnA$ consists of two points (the SMBC$_1$ case), then the problem is to find the shortest path between these two points, a problem famously in P. 

The process of replacing a trail with the shortest path does not in general give a minimum bounding chain (\cref{FIG: Part3 MBC1 section}). However, we can in general use the Folyd-Warshall algorithm \cite{floyd} to get a matrix of the distances between every pair of vertices in the graph, and from this form a new (complete) graph whose vertices are the elements in $\ChnA$ and where the edge weights are the length of the shortest path between the two end vertices. A solution to the MBC$_1$ problem then corresponds to a minimum weight perfect matching in this new graph by \cref{mbc1disjointpaths}, and we can find such a matching in polynomial time using Edmonds blossom algorithm \cite{edmonds_1965}. Thus we have proved \cref{thm:mbc1-in-P}.

\begin{figure}[!h]
\centering
\includegraphics[width = \textwidth]{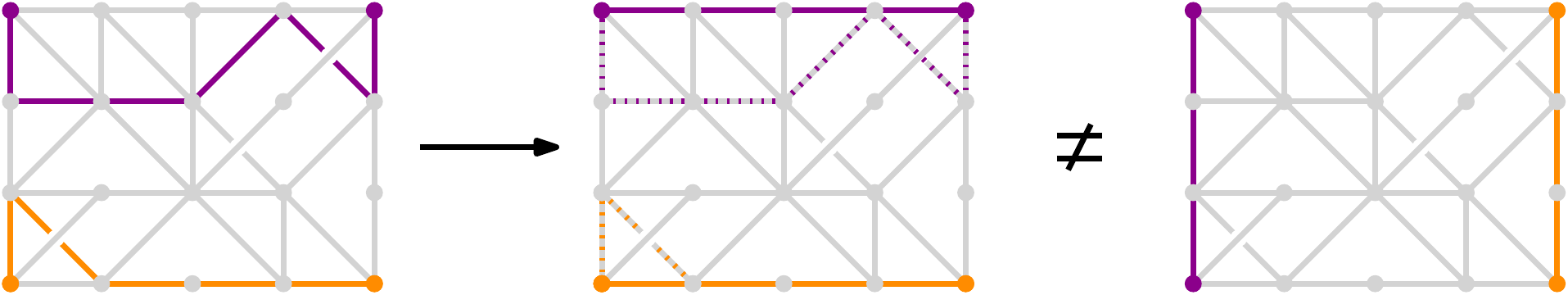}
\caption{\label{FIG: Part3 MBC1 section}A partitioning of the (rightmost) $1$-chain from \cref{FIG: Part2 MBC1 section} (left) and the $1$-chain we get from replacing the two paths between the pairs of vertices with shortest path between the endpoints (middle). The resulting $1$-chain is not the same as the minimum bounding chain, shown to the right.}
\end{figure}

\section{The Dijkstra Approach}\label{sec:dijkstra}
This section introduces an algorithm inspired by Dijkstra's shortest path algorithm to solve the MBC$_d$ problem. Here we also prove hardness results for solution size and coface degree. 
\subsection{The Algorithm}\label{sec:algorithm}

We now describe an algorithm that solves the MBC$_d$ problem for any $d \geq 1$. We first give the basic idea of the algorithm, which we then improve through some simple modifications.

\begin{definition}\label{gdk}
Let $(K,\cosst)$ be a weighted simplicial complex and let $d\geq 1$ be an integer. The graph $G^{d}(K)$ is the weighted graph with vertex set $C_{d-1}(K)$ and edges $(\ChnA,\ChnB)$ between pairs of chains whenever $\partial(\sigma) = \ChnA\triangle \ChnB$ for some $d$-simplex $\sigma$ in $K$. The length of the edge $(\ChnA,\ChnB)$ is equal to the weight of $\sigma$.  
\end{definition}

Since we have chains with coefficients in $\mathbb{Z}_2$, traversing an edge can therefore be viewed geometrically as taking the symmetric difference of the boundary of the simplex $\sigma$ with the current $(d-1)$-chain $\ChnA$. This means that for every bounding chain $\BndA = \{\sigma_1,\sigma_2,\dots, \sigma_k\} $ of $\ChnA$ there is a trail $p(\ChnA,\vec{0}) = (\sigma_1,\sigma_2,\dots, \sigma_k)$ in the graph $G^{d}(K)$ from $\ChnA$ to the empty-chain/null-chain, $\vec{0}$.

\begin{prop}
The MBC$_d$ problem can be solved by finding the shortest path $p(\ChnA,\vec{0})$ in $G^d(K)$, where $\ChnA$ is the given input boundary.\qed % between our input chain $\ChnA$ and the zero-chain $\vec{0}$. \qed
\end{prop}

If $K$ is a simplicial complex with $|K^{d-1}| = n$ and $|K^d| = m$, then this gives us a graph $G^d(K)$ with $2^n$ vertices (each having degree $m$) and $2^{n-1}\cdot m$ edges. Using Dijkstra with Fibonacci heaps \cite{fredman} we can find such a path in $\mathcal{O}(|E|+|V|\log |V|)=\mathcal{O}(2^{n-1}\cdot m + 2^n \cdot n) = \mathcal{O}(2^n(m+n))$ time. 

This gives us an algorithm that is slightly worse than brute force (since $m < n$ trying all possible chains $\BndA$ takes $\mathcal{O}(2^m \operatorname{poly}(m,n))$-time). Next, we make two simple observations that increase the theoretical (parameterized) runtime significantly.

First, note that if we have $\ChnA = \partial(\BndA)$ then for every simplex $\rho \in \ChnA$ the chain $\BndA$ must contain at least some coface $\sigma$ of $\rho$. Second, the order in which we add $d$-simplices to our path is irrelevant. This means that we only need (directed) edges going from a given chain $\ChnA$ to chains $\ChnA \triangle\partial(\sigma)$ where $\rho\in\partial(\sigma)$ for some choice of $\rho$ in $\ChnA$. We make these remarks into a definition and a theorem:

\begin{definition}\label{gdfk}
Let $(K,\cosst)$ be a weighted simplicial complex, let $d\geq 1$ be an integer, and for every $(d-1)$-chain $\ChnA$ in $K$ fix a simplex $\rho_\ChnA\in \ChnA$. The graph $G^{d}_D(K)$ is the directed weighted graph with vertex set $C_{d-1}(K)$ and edges $(\ChnA,\ChnB)$ between pairs of chains whenever $\partial(\sigma) = \ChnA\triangle \ChnB$ for some $d$-dimensional coface $\sigma$ of $\rho_\ChnA$ (see \cref{FIG: Better graph simple}). The length of the edge $(\ChnA,\ChnB)$ is equal to the weight of $\sigma$.  
\end{definition}
\begin{thm}
The MBC$_d$ problem can be solved by finding the shortest path $p(\ChnA,\vec{0})$ in $G_D^d(K)$, where $\ChnA$ is the given input boundary.\qed % between our input chain $\ChnA$ and the zero-chain $\vec{0}$. \qed
\end{thm}

Note that the graph $G_D^d(K)$ is not uniquely determined from the input of the MBC$_d$ alone, since we have to choose some $\rho_\ChnA$ at every node $\ChnA$ of the graph. This choice can be made arbitrarily or deliberately. It would be interesting to see if we can get better algorithms (at least in practice) by exploring different heuristics we can use when making this choice. 
It is enough to consider the component of the graph containing our input boundary $\ChnA$, but as Dijkstra's algorithm can be run while gradually constructing $G^{d}_D(K)$ we do not need to save the entire graph to memory anyways. 

\begin{figure}[!h]
\centering
 \includegraphics[width=\textwidth]{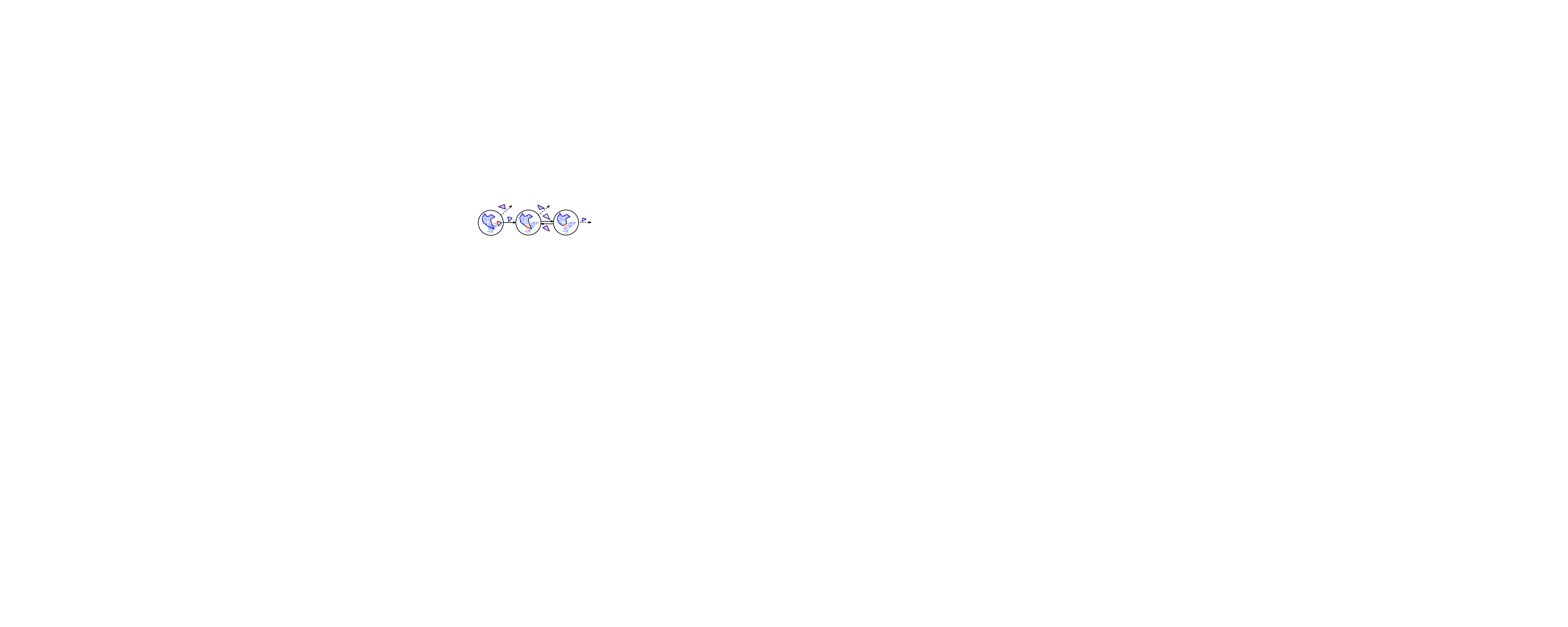}
 \caption{\label{FIG: Better graph simple} A representation of a part of the graph $G^{2}_D(K)$ for some simplicial complex $K$. Nodes in the graph are $1$-boundaries (dark blue) with a fixed $1$-simplex (red). For every $2$-dimensional coface (orange and yellow) of this $1$-simplex we have an edge in the graph, and traversing the edge corresponds to taking the symmetric difference with the boundary of this coface.  }
\end{figure}

The SMBC$_1$ problem when $K$ is a graph corresponds to finding the shortest path $p(s,e)$ between two given nodes $s,e\in K$. If we always pick $\rho_U\neq e$ for the edges in $G^{1}_D(K)$, then we get an isomorphism between the subgraph of $K$ explored by Dijkstra's algorithm finding $p(s,e)$, and the subgraph of $G^{d}_D(K)$ explored finding $p(v+e,\vec{0})$, sending nodes $v+e\mapsto v$ and $\vec{0}\mapsto e$. 

We now work out the runtime we get solving the MBC$_d$ problem by finding the shortest path in $G_D^d(K)$ under different parameterizations. In particular, we investigate the parameters coface degree and solution size. 

The \emph{coface degree} of a $(d-1)$-simplex $\rho$ is defined as the number of cofaces of dimension $d$ (or codimension 1). The \emph{coface degree} (in dimension $d-1$) of a simplicial complex $K$ is the maximum coface degree of any $(d-1)$-simplex in $K$, which we denote by $c$.  In a simplicial complex with coface degree $c$, we know that the graph $G_D^d(K)$ have at most $c$ edges out of every vertex. This means that if we take $c$ to be our parameter, we get a runtime of $\mathcal{O}(2^n(c+n))$ (where we remember that $n$ is the number of $(d-1)$-simplices in $K$). This is still exponential in $n$, so this particular parameterization itself doesn't improve the runtime of our algorithm by much. In \cref{sec: cofacesection} we see that there are good reasons for why our algorithm does not perform better using this parameterization on its own.

The next parameter we want to investigate is the solution size $k$ (i.e. the number of $d$-simplices we want the solution to contain). For the unweighted case, if $k$ is fixed, then we can restrict the search to the $(d-1)$-chains of distance at most $k$ from the input chain $\ChnA$. This can be achieved by keeping track of the size of the path up to each chain, and ignoring them if the path has more than $k$ edges. To estimate the runtime of this parameterized algorithm, we need to count how many vertices there are in this new subgraph. 

First, we see that there is one vertex of distance $0$ from $\ChnA$, namely $\ChnA$ itself. In a chain complex $K$ with coface degree $c$, we can by definition get to at most $c$ new vertices once we are at distance $1$ from $\ChnA$, one for each coface of the $(d-1)$-simplex we choose in $\ChnA$. With each new step we can get to at most $c$ new each for each of the vertices from the step before. Clearly then, there are at most $c^k$ new vertices at distance $k$ from $\ChnA$ than at distance $k-1$.

Thus, if we only want solutions containing at most $k$-simplices, then we only need to consider the subgraph with 
\begin{equation}\label{eq: subgraph}
    |V|=\sum_{i=0}^k c^i = \frac{c^{k+1}-1}{c-1} \quad\text{and}\quad |E| = \frac{c}{2}\cdot\frac{c^{k+1}-1}{c-1}.
\end{equation}
Running Dijkstra on a graph of this size takes 
\begin{equation*}
    \mathcal{O}\left(\frac{\left(c^{k+1}-1\right)\left(2\log\left(\frac{c^{k+1}-1}{c-1}\right)+c\right)}{2(c-1)}  \right)=\mathcal{O}\left(c^k(k\log c + c)\right)\text{-time}.
\end{equation*}

In the weighted case, we need to keep track both of how much the path weighs and how many simplices are on it. A convenient way of dealing with this is by making a slightly larger directed graph and solving Dijkstra there instead. For a general weighted directed graph $G$, consider the directed graph $\overline{G}$ whose vertices are pairs $(v,i)$, consisting of a vertex $v\in V(G)$ and an integer $0\leq i \leq k$. The idea is that the number $i$ keeps count of how many simplices we have added to the path. To make this work, we make the directed edges of $\overline{G}$ be precisely the pairs of vertices $((v,i),(u,i+1))$ where $(v,u)$ is an edge in $E(G)$. The edge $((v,i),(u,i+1))$ is given the same weight as $(v,u)$ had in $G$. We can now use Dijkstra in this graph to find the shortest paths from $(u,0)$ to all the other vertices in this graph. The shortest path from $u$ to $v$ in $G$ is shortest of the distances from $(u,0)$ to any $(v,i)$. This new graph has $|V(\overline{G})| = (k+1)|V(G)|$ and $|E(\overline{G})| = k|E(G)|$.

Running Dijkstra on the graph of size described in \cref{eq: subgraph}, we again get a runtime of the form $c^k\cdot\operatorname{poly}(m, n)$, where the polynomial degree is larger than in the unweighted case.

\begin{thm}\label{PROP: dijkstra runtime}
The MBC$_d$ problem can be solved in $c^k\cdot\operatorname{poly}(m, n)$-time.
\end{thm}
By again noting that $m=\operatorname{poly}(n)$, we have \cref{thm:solution-coface-algorithm}.
In particular, this makes the MBC$_d$ problem solvable in FPT when parameterized by solution size and coface degree. We also get an XP-algortihm for the MBC$_d$ from this analysis, if we take the problem to be parameterized by solution size alone. The XP-runtime is obtained by replacing $c$ with it's worst possible value, namely $m$. This is roughly the runtime we would get if we just tried all the $m^k$ $d$-chains $\BndA$ containing less than $k$-simplices. We show evidence in \cref{sec:solutionsize} indicating that no algorithm parameterized by solution size alone can solve the MBC$_d$ problem in FPT-time for $d\geq2$.

\subsection{Hardness Results}\label{sec:hardness}

In the previous section, we showed that we could solve the MBC$_d$ problem parameterized by coface degree and solution size in FPT-time using our Dijkstra inspired algorithm. The same algorithm turns into an XP-algorithm when we parameterized by solution size alone. If we parameterize the MBC$_d$ problem by coface degree alone the runtime is exponential. Our analysis shows that the algorithm could potentially need exponential time even on spaces where the coface degree is at most 3.

In this section, we use parameterized complexity theory to show that these runtimes actually make a lot of sense. We also show bounds as to how much more we can expect to improve them.

\subsubsection{Coface Degree}\label{sec: cofacesection}

We begin by looking at the parameter maximum coface degree. Recall that for an instance of MBC$_d$, namely a simplicial complex $K$ and a $(d-1)$-chain $\ChnA$, the \emph{maximum coface degree} is the highest number of $d$-simplices in $K$ that share a common $(d-1)$-simplex as a face. If this parameter is $2$, the problem is solvable in polynomial time by a simple preprocessing routine. However, this is the only case where we gain anything, as we have the following result.

\begin{thm}
The unweighted SMBC$_d$ problem for $d\geq 2$ is NP-complete, even when restricted to spaces with coface degree three. 
\end{thm}
In dimension 2, this is not stated in \cite{borradaile_et_al:LIPIcs:2020:12179}, but it follows from their the polynomial time reduction as the output space has maximum coface degree $3$. The reduction in \cref{sec:eth-tightness} can also be altered to show this. To get the result in higher dimensions, we take the suspension like they do in \cite{chen2011hardness}, which does not change the coface degree.
As a consequence we have that the MBC$_d$ problem is NP-complete even for spaces with coface degree three, as stated in \cref{thm:coface-hardness}.

\subsubsection{Solution Size}\label{sec:solutionsize}

The next parameter we look at is the \emph{solution size}, which is how many simplices there are in the solution. Be aware that we do not talk about the \emph{solution weight}, that is the sum of the weights of these simplices. We begin by stating the hardness result.

\begin{lemma}\label{lem: smbcd W1}
The unweighted SMBC$_d$ problem is W[1]-hard when parameterized by solution size. 
\end{lemma}
We show the $d=2$ case, and the higher dimensions follow again by taking the suspension, which changes the solution size linearly by doubling once for each suspension.
It may be possible to prove this result with a modified argument based on the parameterized reduction from the \textsc{Grid Tiling} problem to the \textsc{2-Sphere Recognition} problem in \cite{Burton2019}. Here, we present a completely different reduction from the $\alpha\times\beta$-\textsc{Clique} problem defined below, as this gives us further hardness results for when the MBC$_d$ problem is parameterized with respect to both solution size and maximum coface degree (see \cref{sec:doubleparameter}).

Before we state the problem, we define a $\alpha\times\beta$-\emph{grid graph} $G$ for positive integers $\alpha,\beta$ to be a graph where each vertex is uniquely specified by two numbers, its column $i$ (where $1 \leq i\leq \alpha$) and row $j$ (where $1 \leq j \leq \beta$). We denote vertices by $(i,j)$, and to keep track of what happens in the reduction, we give the vertices of each column the same color, saying that the vertex $(i,j)$ has color $i$. We draw $\alpha\times\beta$-grid graphs as in \cref{FIG: grid graph example}, where we place the vertices in rows and columns forming a grid. An $\alpha$-\emph{clique} in an $\alpha\times\beta$-grid graph is a collection of $\alpha$ different colored vertices all having edges between each other.

\begin{figure}[!h]
\centering
 \includegraphics[width=0.9\textwidth]{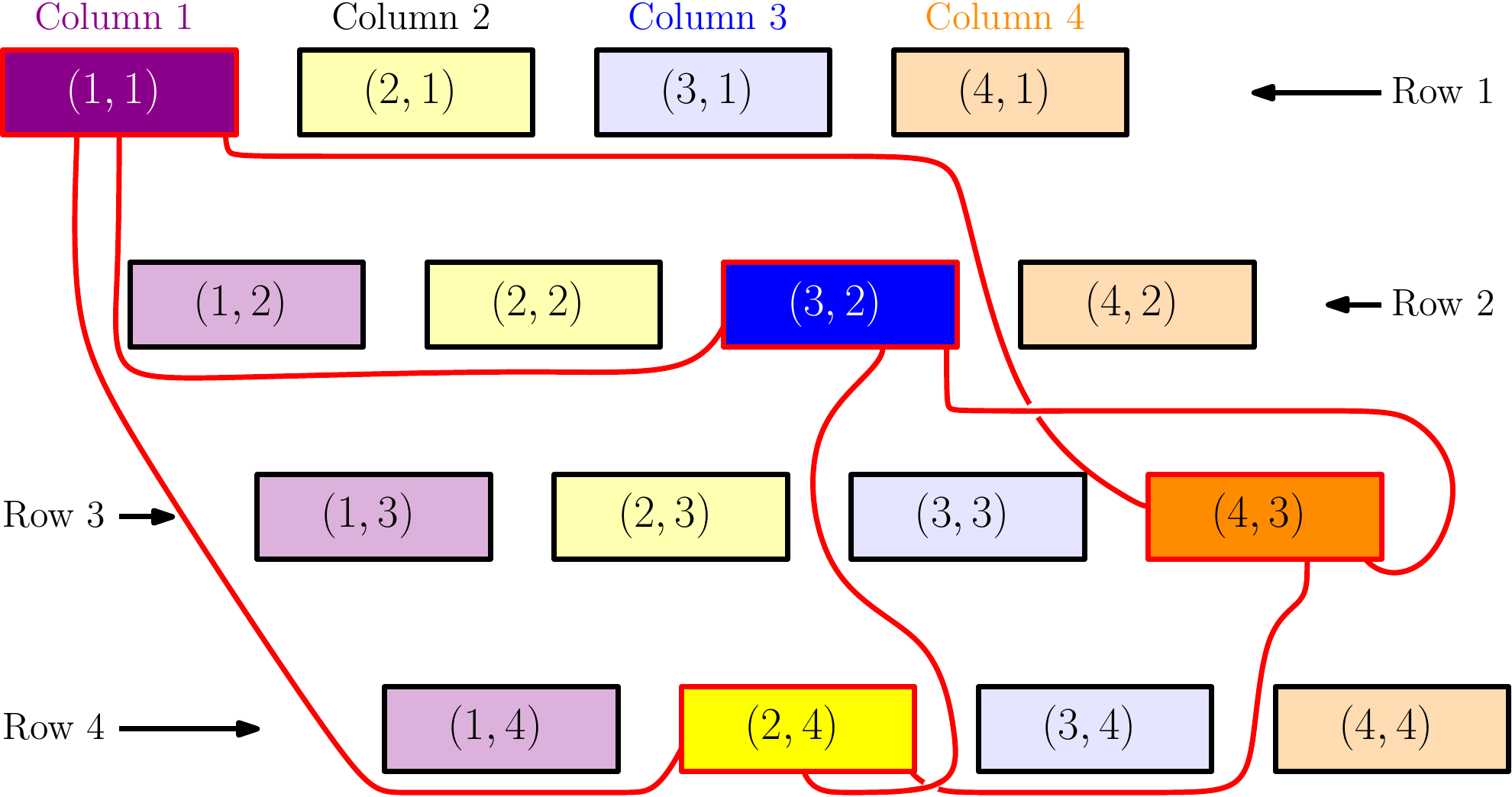}
 \caption{\label{FIG: grid graph example} An example of a $4\times 4$-grid graph (the vertices are drawn as boxes) containing a $4$-clique on the $4$ emphasized vertices of different colors/columns.}
\end{figure}

\begin{definition} The $\alpha\times\beta$-\textsc{Clique} problem\\
INPUT: A $\alpha\times\beta$-grid graph $G$.\\
QUESTION: Is there a $\alpha$-clique in $G$ containing one vertex from each column?
\end{definition}

The $\alpha\times\beta$-\textsc{Clique} problem is W[1]-hard when parameterized by \(\alpha\) \cite[Lemma 1]{fellows2007fixed}.

We now describe a polynomial reduction from the $\alpha\times\beta$-\textsc{Clique} problem to the SMBC$_2$ problem. The first step is to give a polynomial time algorithm for constructing simplicial complexes $X(G)$ from any given $\alpha\times\beta$-grid graphs $G$. This space $X(G)$ is made up of basic building blocks pictured in \cref{FIG: Notation}, which also shows the short hand notation we use throughout this section.

The space $X(G)$ is rather complicated, so we split the construction into five ``layers'', each consisting of one or more copies of the objects introduced in \cref{FIG: Notation}. In \cref{FIG: General reduction with layers} we have a systematic overview of which objects occupy any given layer and how these layers are connected. For a concrete example of the reduction see \cref{FIG: W1 hardness reduction example}, showing the space $X(G)$ where $G$ is the grid graph from \cref{FIG: grid graph example}. Finally, to see the location of the minimum bounding chain corresponding in $X(G)$ that corresponds to the clique in $G$, see \cref{FIG: concrete example}.

\begin{figure}[!h]
\centering
 \includegraphics[width=\textwidth]{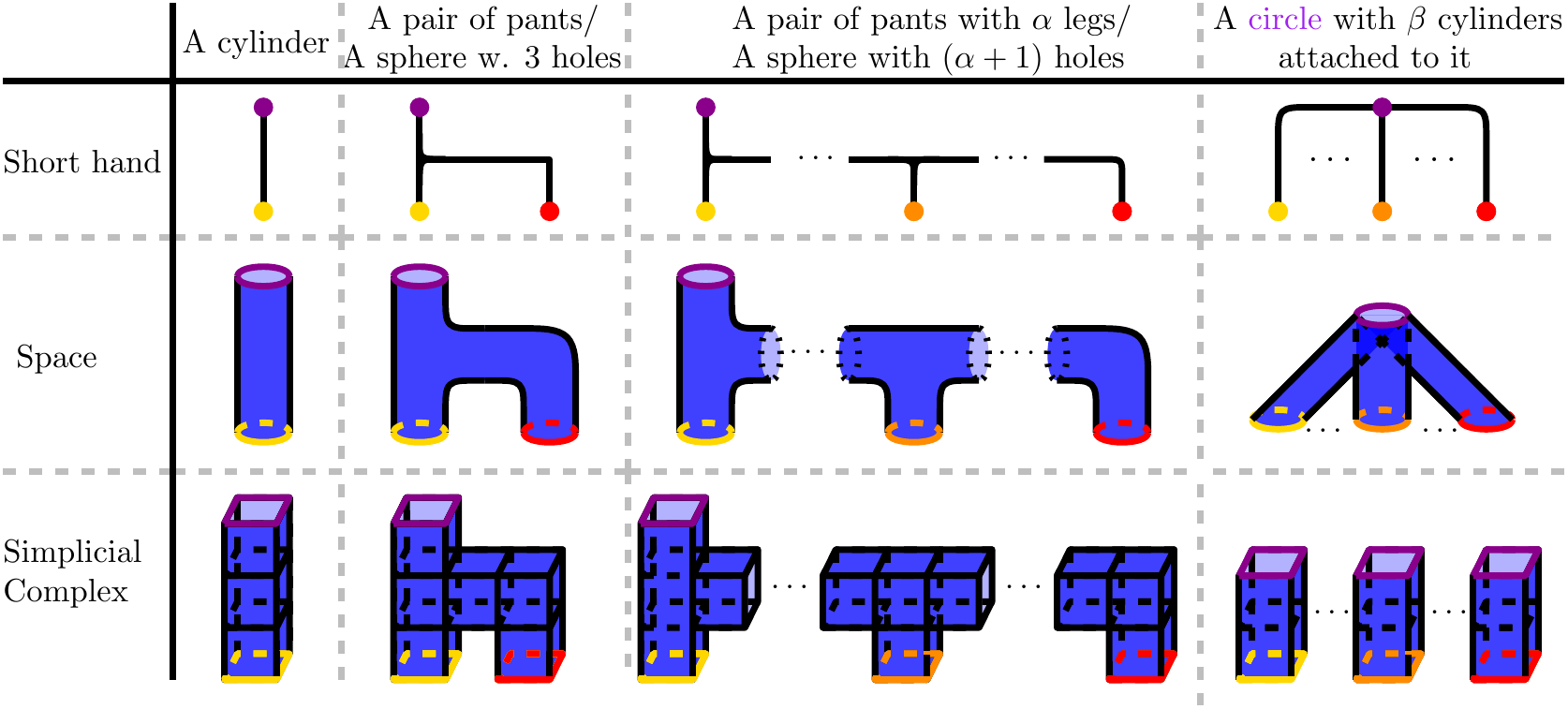}
 \caption{\label{FIG: Notation} An overview of short hand notation used in several of the reductions in this paper. We glue these spaces together along their boundaries when we combine them.}
\end{figure}

The first layer of the space $X(G)$ is simple: it is a pair of pants with $\alpha$ legs, as pictured in \ref{FIG: Notation}. The input chain of the SMBC$_2$ problem is the $1$-simplices on the boundary of the ``waist'' of the pair of pants, which we denote by $\Chnpic$. The boundary of each leg represents precisely one of the $\alpha$ colors in $G$ and we denote them by $x_1,\cdots x_k$. %The first layer is pictured in \cref{FIG: General reduction with layers}.

The second layer of the reduction is given by gluing $\beta$ cylinders to each leg $x_i$ along one of the boundary components, where we recall $\beta$ is the number of vertices of each color. We denote the cylinders boundary component that is not glued to $x_i$ by $y_{i,j}$ where $1\leq j\leq \beta$, as seen in \cref{FIG: General reduction with layers}. Here $y_{i,j}$ corresponds to the vertex $(i,j)$ in our graph $G$, and if the cylinder to $y_{i,j}$ is part of a solution to the SMBC$_2$ problem in $X(G)$ with boundary $\Chnpic$, then the vertex $(i,j)$ is part of an $\alpha$-clique in $G$.

With this second layer we have represented the all the vertices of the graph $G$, and we need to encode the edges. The most naive way would be to glue a cylinder to the boundaries $y_{i,j}$ and $y_{i',j'}$ if there is an edge between the vertices the two boundaries represent. However, this encoding does not work as seen by the counterexample in \cref{FIG: General reduction with layers}.

Instead, we encode the fact that every vertex must have $\alpha-1$ neighbours, each one of a different color. In particular, we want to say that each cycle representing a vertex must be canceled $\alpha-1$ times, once for every color different from its own. The first step (layer 3) is to attach a pair of pants with $\alpha-1$ legs (similar to the one in layer 1) to each circle $y_{i,j}$, one leg for each remaining color. We denote the cycle at the boundary of each leg by $y^{i,j}_{i'}$, where $1 \leq i'\leq \alpha$ and $i'\neq i$ (see \cref{FIG: General reduction with layers}). 
Layer 4 is another layer of cylinders, similar to layer 2. We glue $\beta$ cylinders to each $y^{i,j}_{i'}$, representing the possible vertices of color $i'$ that share an edge with the vertex represented by $y^{i,j}$. In layer 2 we forced a solution to pick a vertex in every color, and in layer 4 we force it for each chosen vertex to pick a neighbor of that vertex of every other color. We denote the new boundaries created by $z^{i,j}_{i',j'}$, where $i,\,j$ and $i'$ is as above and $1\leq j'\leq \beta$ represents a possible neighbour of the vertex represented by $y_{i,j}$ of color $i'$.

Finally, in layer 5, we encode the edges. If there is an edge in the graph $G$ from vertex $j$ of color $i$ to vertex $j'$ of color $i'$, then $z^{i,j}_{i',j'}$ and $z_{i,j}^{i',j'}$ are connected by a cylinder $edge((i,j),(i'j'))$. 

\begin{figure}[!h]
\centering
  \includegraphics[width=\textwidth]{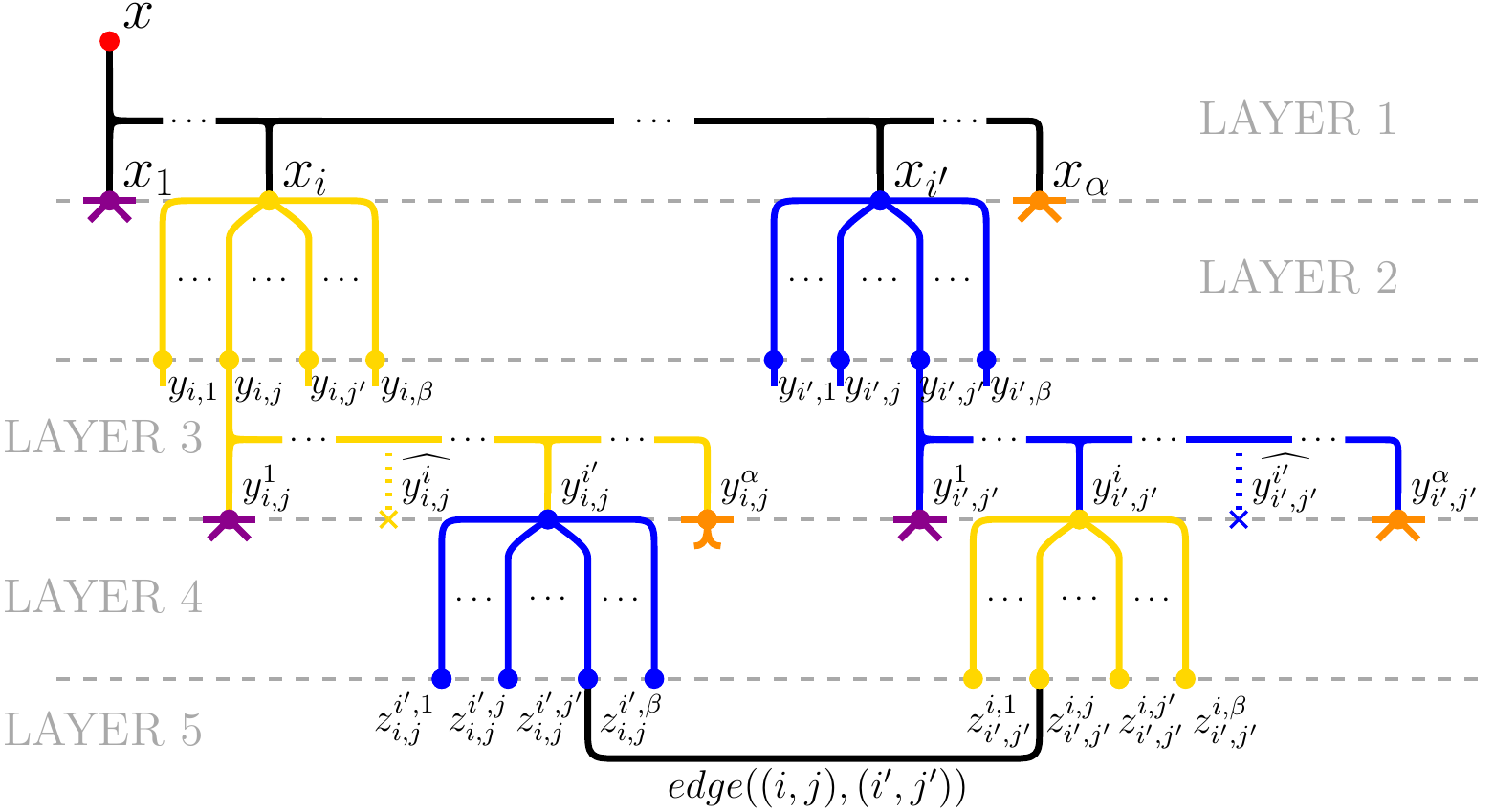}
 \caption{\label{FIG: General reduction with layers} A small portion of each of the five layers in the reduction using the short hand notation from \cref{FIG: Notation}.}
\end{figure}

The problem of finding any bounding chain for $\Chnpic$ in this space may seem to be equivalent with finding a clique in the input graph. However, we know that this is not the case, as it is still easy to look for such a bounding chain using linear algebra. We need to ask if there is a bounding chain in this complex of small size. To find the exact value of the bounding chain we are looking for, requires some careful counting. Roughly speaking, for some constants $A_0,\,A_1$ and $A_2$, we want a solution of size equal to the sum of:

\begin{enumerate}
    \itemsep0pt
    \item the 2-simplices in the pair of pants with $\alpha$ legs pictured in \cref{FIG: Notation}, given by the formula $A_1\cdot\alpha + A_0$.
    \item the number of 2-simplices in a cylinder multiplied by $\alpha$, given by the formula $A_2 \cdot \alpha$.
    \item the number of 2-simplices in the pair of pants with $\alpha-1$ legs multiplied by $\alpha$,  given by the formula $(A_1(\alpha-1) + A_0)\cdot \alpha$.
    \item the number of 2-simplices in a cylinder multiplied by $\alpha(\alpha-1)$,  given by the formula $A_2\cdot\alpha(\alpha-1)$.
    \item the number of 2-simplices in a cylinder $edge((i,j),(i',j'))$ multiplied by $\alpha(\alpha-1)/2$,  given by the formula $A_2\cdot\alpha(\alpha-1)/2$.
\end{enumerate}

Let $k' = k'(\alpha)$ be given by the formula
\begin{align*}
    k'(\alpha)  &=A_1\alpha + A_0 +  A_2 \alpha + (A_1(\alpha-1) + A_0)\alpha + A_2\alpha(\alpha-1) + A_2\alpha(\alpha-1)/2\\
                &= A_2'\alpha^2 + A_1'\alpha + A_0'
\end{align*}
for constants $A_2',A_1',A_0$.

\begin{figure}[!h]
\centering
\begin{subfigure}{.47\textwidth}

\vspace{0.9cm}

\centering
\includegraphics[width=\textwidth]{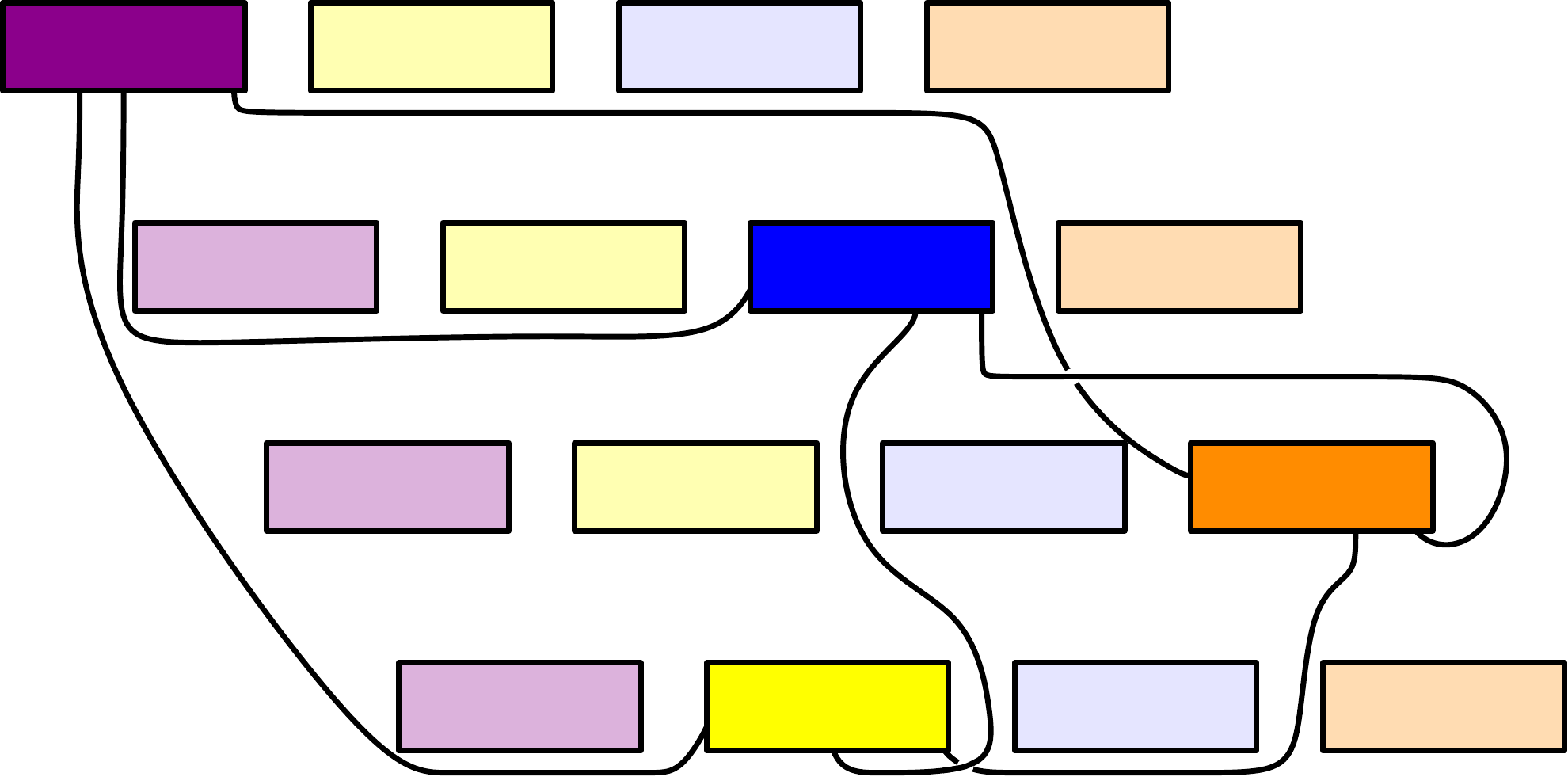}
\end{subfigure}
\begin{subfigure}{.52\textwidth}
\centering
\includegraphics[width=\textwidth]{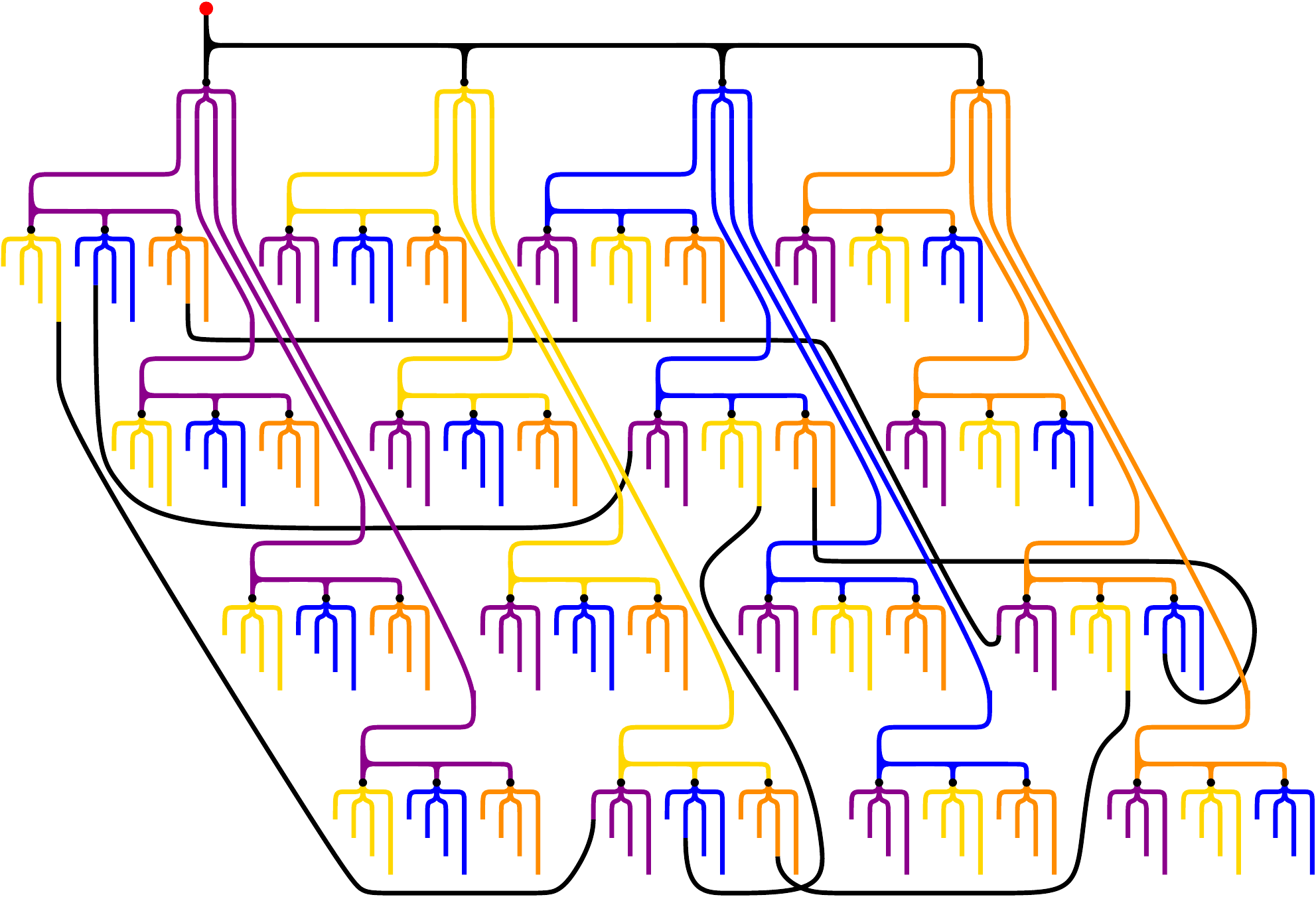}
\end{subfigure}
\caption{\label{FIG: W1 hardness reduction example} A concrete example of the reduction, showing the $4\times 4$-grid graph $G$ from \cref{FIG: grid graph example} together with the space $X(G)$.}
\end{figure}

\begin{prop}\label{PROP: Second reduction works}
Mapping $(G,\alpha)$ to $(X(G),\Chnpic ,k'(\alpha))$ is a parameterized reduction from the $\alpha\times\beta$-\textsc{Clique} problem parameterized by $\alpha$ to the SMBC$_2$ problem parameterized by solution size $k'$.
\end{prop}

\begin{proof}

The reduction runs in polynomial time, as the size of the output instance is polynomial in the size of the input instance. So we have yet to show that there exists a solution of the $\alpha\times\beta$-\textsc{Clique} problem if and only if there exists a solution to the SMBC$_2$ problem of weight less than or equal to $k'$.

From the observations made during the construction we know that if there is a clique $\{(1,j_1), \cdots ,(\alpha,j_\alpha)\}$ in $G$ then there is a bounding chain $\BndA$ of $\Chnpic$ in $X(G)$ of weight $k'$. Explicitly, the chain $\BndA$ consists of all simplices in 
\begin{itemize}
    \item the pair of pants in layer 1,
    \item for $1\leq i \leq \alpha$ the $\alpha$ cylinders ending in $y_{i,j_i}$ in layer 2 and their corresponding pair of pants in layer 3,
    \item for each of these pair of pants with waist $y_{i,j_i}$ the $\alpha-1$ cylinders ending in $z_{i,j_i}^{i',j_{i'}}$ in layer 4,
    \item and all the connecting cylinders between $z_{i,j_i}^{i',j_{i'}}$ and $z^{i,j_i}_{i',j_{i'}}$ in layer 5.
\end{itemize}

For the converse implication, we let $\BndA$ be a bounding chain of $\Chnpic$ in $X(G)$ of size $k'$, and look at which simplices must be part of $\BndA$ whenever $\partial\BndA = \Chnpic$. For $\Chnpic$ to be the boundary, the chain $\BndA$ must contain all the simplices in the pair of pants in layer 1. This pair of pants has an extra boundary consisting of the circles $x_i$ for $1\leq i \leq \alpha$, which needs to be canceled by some other simplices in $\BndA$. To cancel $x_i$ the bounding chain $\BndA$ needs to contain all simplices of an odd number of cylinders intersecting $x_i$ in layer 2. In particular it needs to contain at least one such cylinder, giving a new boundary $y_{i,j}$ which can only be canceled by adding the corresponding pair of pants in layer 3. This gives $\alpha-1$ new boundaries $y_{i,j}^{i'}$ for $i'\neq i$, each of which we again need cancel by adding the simplices in at least one cylinder in layer 4 for each $i'$. After this process we are left with at least $\alpha(\alpha-1)$ extra boundaries $z_{i,j}^{i',j'}$ that still have to be canceled by simplices in the chain $\BndA$, and even more if we pick more than one cylinder in layer 2 and 4. 

Each $z_{i,j}^{i',j'}$ has to be removed by a cylinder in layer 5, and a cylinder can remove at most two such boundaries. Since we know that the maximum of the sum of weights of simplices in $\BndA$ is $k'$, and by subtracting the weights of the simplices we already know is part of it, we see that there is only room for at most $\alpha(\alpha-1)/2$ more cylinders in $W$. This is exactly the minimum amount to cancel the rest of the boundaries, so we know they must be part of the bounding chain $\BndA$ and that $\BndA$ does not contain any other simplices. Thus we conclude that only one cylinder is picked for each $x_i$, the one corresponding to some vertex $(i,j_i)$. The collection of these vertices $\{(i,j(i)) | 1 \leq i \leq \alpha\}$ forms a clique as there must be cylinders in $W$ and therefore in the space $X(G)$ connecting every pair of boundaries $(z_{i,j(i)}^{i', j(i')}, z^{i,j(i)}_{i', j(i')})$. These cylinders are present in $X(G)$ if and only if there are edges in $G$ going between the vertex $j(i)$ of color $i$ and vertex $j'(i')$ of color $i'$ in $G$, so we have our result.

\end{proof}

This finishes the proof of \cref{lem: smbcd W1}, and therefore of \cref{thm:solution-hardness}.

\begin{figure}[!h]
\includegraphics[width=\textwidth]{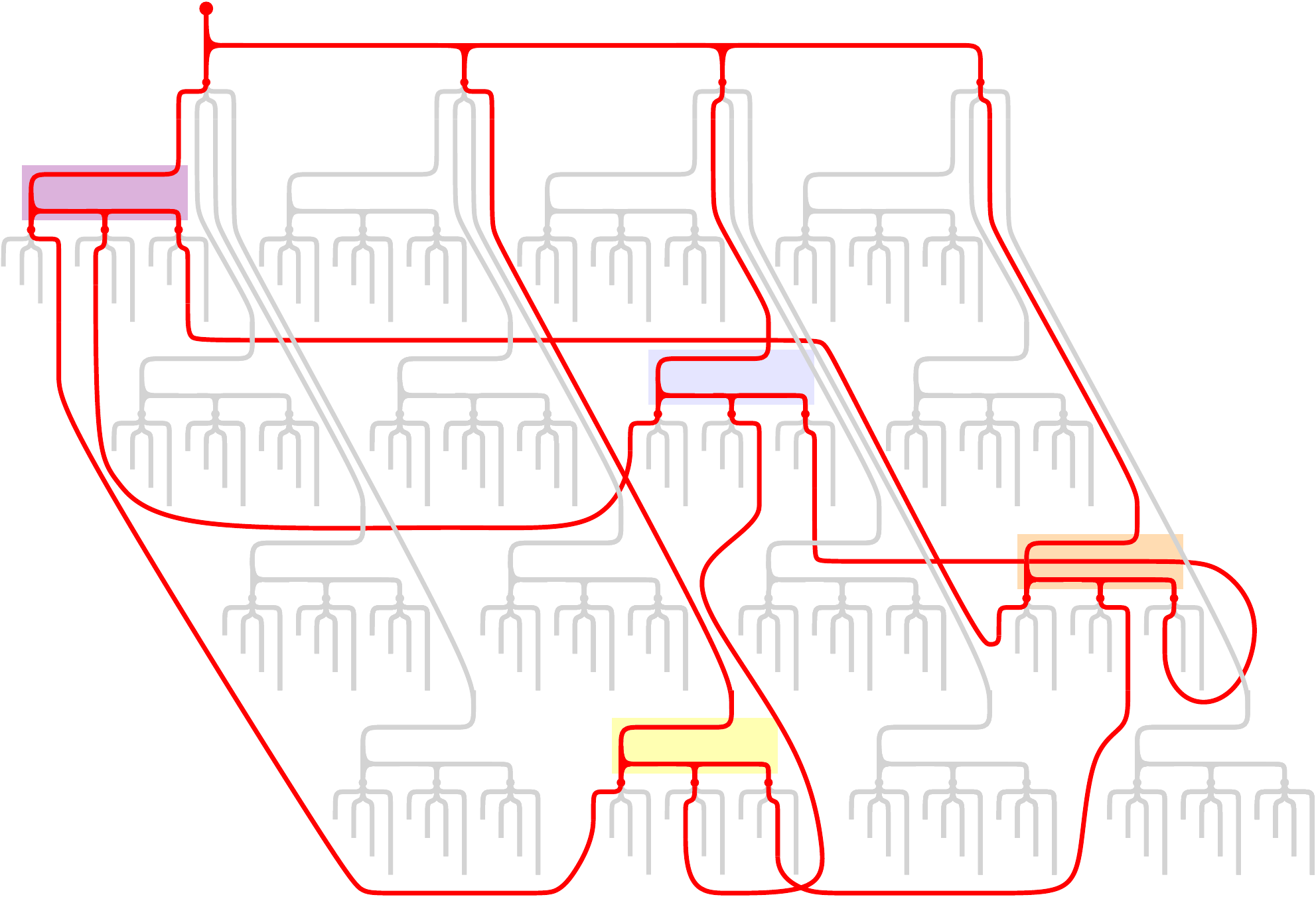}
\caption{\label{FIG: concrete example} The minimum bounding chain in the space $X(G)$ from \cref{FIG: W1 hardness reduction example} where the input boundary is the $1$-simplices in the topmost circle/waist. Note that it picks out every cylinder in layer 5 corresponding to edges in the clique.}
\end{figure}

\subsubsection{Solution Size and Coface Degree}\label{sec:doubleparameter}

We saw in \cref{sec: cofacesection} that parameterizing the MBC$_d$ problem with respect to the coface degree still yields an NP-complete problem. In \cref{sec:solutionsize} we saw that restricting the solution size does help a bit, but the problem is still W[1]-hard. We also saw in \cref{sec:algorithm} that the problem is polynomial when parameterized by both solution size and coface degree. In this section we give an ETH based hardness result giving a lower bound on runtime when considering both of these parameters at once. 

Before we turn to this theorem we need a lemma from the parameterized complexity theory ``folklore''. For completeness we have included a proof based on a sketch obtained in private correspondence with Daniel Lokshtanov.

\begin{lemma} \label{LEMMA: Daniel L}
The $\alpha\times\beta$-\textsc{Clique} problem can not be solved in $2^{o(\alpha\log(\beta))}$-time, assuming the ETH. 
\end{lemma}

\begin{proof}
We make a small alteration to the $k\times k$-clique result form \cite[Thm. 14.12]{PA}, which gives a reduction from the $3$\textsc{-Coloring} problem. 

Let $G$ be a graph with $|V(G)|=N$ vertices, and let $V_1,\dots,V_\alpha$ be a cover of $V(G)$ where $|V_i|\leq N/\alpha$ for every $1\leq i\leq \alpha$. There is at most $3^{|V_i|}\leq 3^{N/\alpha}$ 3-colorings of the full subgraph of $G$ with vertices $V_i$, so let $\beta=3^{N/\alpha}$ and let $\gamma^1_i,\cdots,\gamma^\beta_i$ be all such colorings (possibly duplicated) of $V_i$. Define an $\alpha\times\beta$-grid graph $H$ whose vertices $(i,j)$ and $(i',j')$ have an edge between them whenever $i\neq i'$ and $\gamma^i_j\cup\gamma^{i'}_{j'}$ is a valid coloring of the full subgraph of $G$ with vertices $V_i\cup V_{i'}$. 

A 3-coloring $\gamma$ of $G$ corresponds to a clique $\{(i,j)\,|\,\gamma^i_j=\gamma|_{V_i}\}$ in $H$, and conversely a clique $\{(i,j_i)\}$ in $H$ corresponds to a coloring $\gamma=\bigcup_{i=1}^\alpha \gamma^i_{j_i}$ of $G$. If we can solve the $\alpha\times\beta$-\textsc{Clique} problem in $2^{o(\alpha\log(\beta))}$-time, then we can solve the $3$\textsc{-Coloring} problem in 
\begin{equation}
    2^{o(\alpha\log(\beta))} = 2^{o(\alpha\log(3^{N/\alpha}))}= 2^{o(N)}\text{-time}
\end{equation}
which contradicts the ETH.
\end{proof}

\begin{lemma}\label{lem: eth-hardness double}
Assuming the ETH, the unweighted SMBC$_2$ problem can not be solved in $2^{o(\sqrt{k}\log(c))}\operatorname{poly}(n)$-time. 
\end{lemma}

\begin{proof}
This is a consequence of \cref{LEMMA: Daniel L} combined with \cref{PROP: Second reduction works}. 

The proof involves finding a contradiction with the ETH. To this end, assume that there exists some parameterized algorithm $A$ solving the unweighted SMBC$_2$ problem in $2^{o(\sqrt{k}\log(c))}\operatorname{poly}(n)$-time. Given some problem instance $(G,\alpha,\beta)$ of the $\alpha\times \beta$-\textsc{Clique} problem, we can reduce it in polynomial time to an instance $(X'(G),k'(\alpha),\beta)$ of the SMBC$_2$ problem. We can now use algorithm $A$ to find a solution to this instance in 
\begin{align*}
    2^{o(\sqrt{k'(\alpha)}\log(\beta))}\operatorname{poly}(|X'(G)|) &= 2^{o(\sqrt{A_2'\alpha^2 + A_1'\alpha + A_0'} \log(\beta))} \operatorname{poly}( \operatorname{poly} (\alpha, \beta) )\\
    &=2^{o(\alpha\log(\beta))}\operatorname{poly}(\alpha, \beta) \\
    &= 2^{o(\alpha\log(\beta))}\text{-time,}
\end{align*}
In other words, we have an algorithm solving the $\alpha\times \beta$ \textsc{Clique} problem in $2^{o(\alpha\log(\beta))}$-time, which contradicts the ETH.
\end{proof}
The general case MBC$_2$ has to be at least as difficult as the special case SMBC$_2$, and by taking the suspension to higher dimension we have \cref{thm:solution-coface-hardness}.

\begin{remark}
From \cref{thm:solution-coface-hardness} we know that the MBC$_d$ problem cannot be solved in $2^{o(\sqrt{k}\log(c))}\operatorname{poly}(n)$-time (assuming the ETH) and from \cref{thm:solution-coface-algorithm} we know that it can be solved in $2^{\mathcal{O}(k\log(c))}\operatorname{poly}(n)$-time. This leaves an obvious gap, and it is not clear how to bridge it.

There might be a better parameterized algorithm than the one presented in this paper. In particular, we have the ETH-tight $2^{\mathcal{O}(k)}n^{\mathcal{O}(\sqrt{k})}$-time algorithm from Theorem 2 of B. Burton et al. \cite{Burton2019} that can recognize if a simplicial complex of size $n$ contains a $2$-sphere of size (at most) $k$ as a sub-complex. This is interesting, because this problem is similar to the MBC$_2$ problem while also having a runtime close to what we are aiming for.

The square root emerges as a consequence of the fact that the treewidth of the underlying graph of any triangulation of a sphere using $k$ simplices is in $\mathcal{O}(\sqrt{k})$. While this is also the case for surfaces of fixed genus it is not true for $2$-chains in general. It seems therefore unlikely that a similar technique can be used on the MBC$_d$ problem.

Finally, there might be some other reduction that gives a better (i.e. higher) lower bound. Note that because of the algorithm by B. Burton et al., such a reduction needs to have certain properties, assuming the ETH is true. In particular, the optimal solutions to the instances in the image of the reduction can not all be surfaces of bounded genus.
\end{remark}

\section{Treewidth}\label{sec:treewidth}

We study the parameterized complexity of the MBC$_d$ problem parameterized by $\tau$, the treewidth of the $d$'th level of the Hasse diagram, which we describe later. This section contains two main results: 
\begin{enumerate}
\itemsep0pt
    \item The MBC$_d$ problem can be solved in $\mathcal{O}(2^{2\tau} \tau^2 n)$-time when parameterized by $\tau$.
    \item This algorithm is ETH-tight for $d\geq 2$ (no $2^{o(\tau)}\operatorname{poly}(n)$-time algorithm exists unless the ETH is false).
\end{enumerate}

These results share many similarities with the main results of \cite{erlend_homloc} concerning the related \textsc{Homology Localization} problem, and several details are the same. For this reason, we leave parts of the proofs to that paper, where the analogous proofs are given in great detail.

\subsection{Tree Decompositions} \label{sec: treedecomposition of spaces}

A \emph{tree} $T$ is a connected graph with no cycles (i.e. $H_1(T)$ is trivial). A \emph{rooted tree} $(T,r)$ is a tree $T$ together with a vertex $r\in V(T)$ called the \emph{root}. A vertex $s$ in $(T,r)$ is said to be the \emph{descendant} of another vertex $t$ if $t$ appears on the (unique) path from $s$ to $r$. If $t$ is the first vertex on this path, then $t$ is the \emph{parent} of $s$ and $s$ is a \emph{child} of $t$. Vertices with no children are called \emph{leaves}.

Intuitively, the treewidth is a measure of how close a given graph is to being a tree (see \cref{FIG: TD of graphs}). Many problems become solvable in FPT-time when they are parameterized by treewidth in the same way that many NP-complete problems become solvable in polynomial time when we restrict the input graphs to be trees. We define treewidth in terms of tree decompositions of graphs in this paper. When we design our algorithm, which is a dynamic programming routine on a tree decomposition of a graph.

\begin{definition}[Nice Tree Decomposition]
A \emph{tree decomposition} of a graph $G$ is a rooted tree $(T,r)$ together with a function $X_-:V(T)\to\mathcal{P}(V(G))$ mapping vertices $t$ in $T$ to subsets $X_t\subseteq V(G)$ called \emph{bags}. This map must have the following properties:

\begin{itemize}
\itemsep0pt
    \item For all vertices $v$ in $G$ there exists a vertex $t$ in $T$ such that $v\in X_t$.
    \item For all edges $vu$ in $G$ there exists a vertex $t$ in $T$ such that $u, v\in X_t$.
    \item If $u\in X_t\cap X_{t'}$ for vertices $t,t'$ in $T$ then $u\in X_{s}$ for every vertex $s$ on the path in $T$ from $t$ to $t'$. 
\end{itemize}
A tree decomposition is said to be \emph{nice} if $X_r = \emptyset$ and every bag $X_t$ is one of the following:
\begin{itemize}
\itemsep0pt
    \item A \emph{leaf bag} where $t$ is a leaf and $X_t = \emptyset$.
    \item An \emph{introduce bag} where $t$ has a child $s$, and $X_t = X_s\sqcup \{v\}$ for a vertex $v$ in $G$.
    \item A \emph{forget bag} where $t$ has a child $s$, and $X_t\sqcup \{v\} = X_s$  for a vertex $v$ in $G$.
    \item A \emph{join bag} where $t$ has two children, $s$ and $s'$, and $X_t = X_s = X_{s'}$.
\end{itemize}
\end{definition}

\begin{figure}[!h]
\centering
\includegraphics[width=0.45\textwidth]{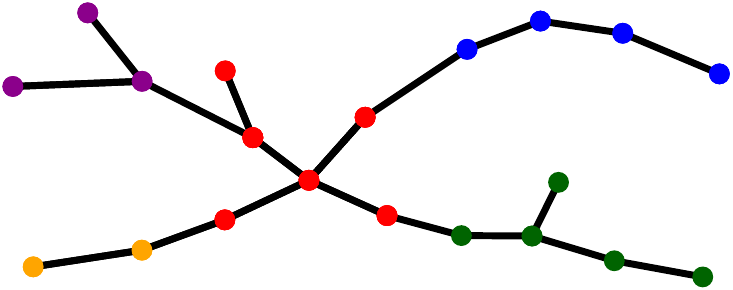}
\hspace{0.5cm}
\includegraphics[width=0.45\textwidth]{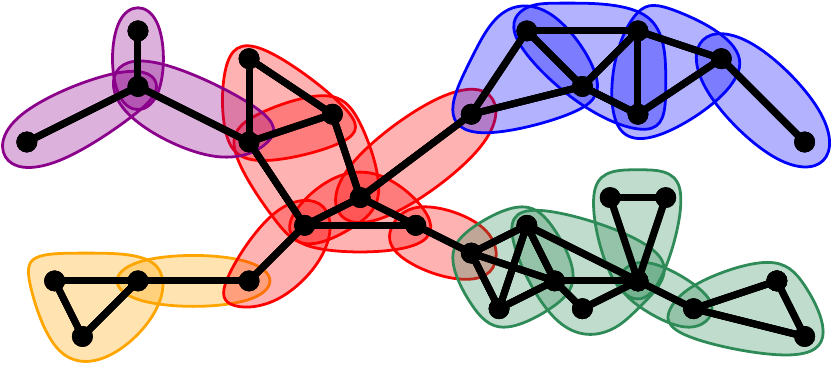}
\caption{\label{FIG: TD of graphs} A figure illustrating a tree decomposition. The figure to the left shows the tree and the figure to the right shows a graph covered by sets (i.e. bags) of different colours. The nodes in the figure to the left maps to the sets covering the graph in such a way that nodes are mapped to bags of similar colour and relative position. The width of this decomposition (and also the treewidth of the graph), is $3 = 4-1$.}
\end{figure}

\begin{definition}[Treewidth] 
The \emph{width} of a tree decomposition is the size of the largest bag it contains minus one. The \emph{treewidth} of a graph is the smallest width of all the possible tree decompositions of that graph.
\end{definition}

It is well known that every tree decomposition of a graph $G$ can be transformed into a nice tree decomposition of $G$ without increasing the width and while keeping the number of bags it contains linear in $|V(G)|$.

\subsection{Graph Maximum Likelihood Decoding}

We present an FPT-algorithm for the \textsc{MBC}$_d$ problem parameterized by the treewidth of the $d$'th level of the Hasse diagram of the simplicial complex. In fact, the algorithm we describe is more general, as it also solves the \textsc{MLD} problem in FPT-time where we use the treewidth of the bipartite graph $\HasseGraph(A)$ as a parameter.

%Our first step is to transform the \textsc{MLD} problem into a problem about graphs.
\begin{definition}\label{Def: Hassegraph of a matrix}
We can represent any matrix $A$ with coefficients in $\Zto$ as a bipartite graph $\HasseGraph(A)$ where the rows $\rho_i$ and columns $\sigma_j$ of $A$ are vertices and where the edges are pairs of rows and columns $(\rho_i,\sigma_j)$ such that $A_{i,j}= 1$.
\end{definition}

Another way of defining $\HasseGraph(A)$ is to say that it is the bipartite graph having $A$ as its biadjacency matrix. See \cref{FIG:graphs_in_matrices} for a small example of what $\HasseGraph(A)$ may look like. 
\begin{figure}[!h]
\centering
\includegraphics[width=0.8\textwidth]{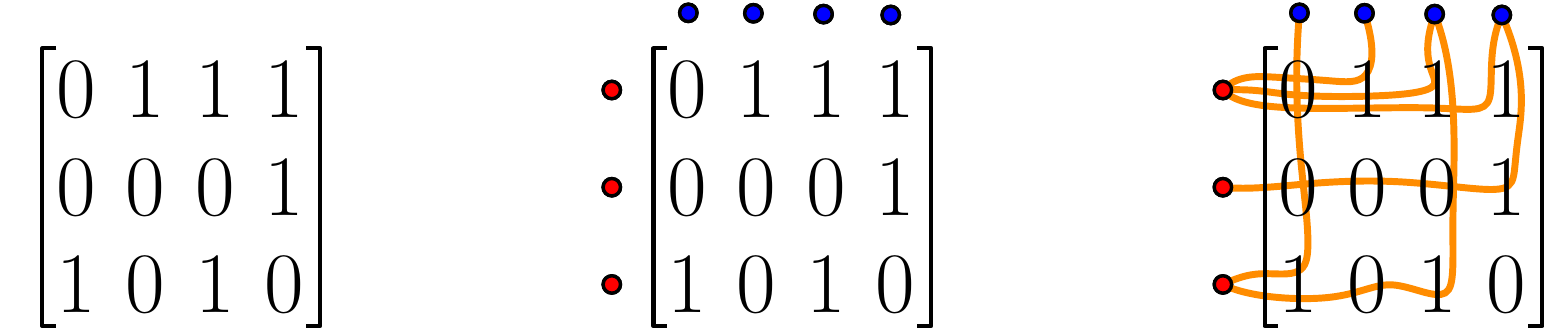}
\caption{A matrix $A$ with coefficients in $\Zto$, and the graph $\HasseGraph(A)$ whose vertices are rows (red) and columns (blue), and there are edges (orange) between a row $i$ and a column $j$ if the element $A_{i,j}$ is $1$.  
\label{FIG:graphs_in_matrices}}
\end{figure}

We can now reformulate the \textsc{MLD} problem as a problem on bipartite graphs. Let $G$ be a graph and let $\partial \sigma$ denote the set of neighbors of a vertex $\sigma$ in $G$. The \emph{boundary} $\partial \BndA$ of a subset of vertices $\BndA\subseteq V(G)$ is the symmetric difference of the neighbors of all vertices in $\BndA$, i.e. $\partial \BndA = \triangle_{\sigma\in \BndA}\partial \sigma$.

\begin{definition}\textsc{Graph Maximum-Likelihood Decoding (GMLD)}:\newline
INPUT: A bipartite graph $G$ with vertex set $(R,C)$, a set of weights $\{w_{\sigma} | {\sigma\in C}\}$ and a subset $\ChnA \subseteq R$. \newline
OUTPUT: A subset $\BndA\subseteq C$ where $\partial \BndA = \ChnA$. \newline
MINIMIZE: The weight $\cosst(\BndA)=\sum_{\sigma\in \BndA} w_{\sigma}$.
\end{definition}

\begin{thm}
The MLD problem can be solved in $\mathcal{O}(2^{2\tau}\tau^2n)$-time when parameterized by the width $\tau$ of a (nice) tree decomposition of $\HasseGraph(A)$, which we assume is given as part of the input. 
\end{thm}

\cref{thm:treewidth-algorithm} as an almost immediate consequence of this result. If we are not given a nice tree decomposition as part of the input, we first compute one. There is an algorithm running in $2^{\mathcal{O}(\tau)}n$-time that finds a tree decomposition whose width \(\tau'\) is a constant factor approximation of the true treewidth \(\tau\) \cite{bodlaender2016c}.

To see how solving the GMLD problem can be used to solve the MBC$_d$ problem, let $(K,\ChnA)$ be an instance of the MBC$_d$ problem consisting of a simplicial complex $K$ and a boundary $\ChnA$. To reduce to the MLD problem, set $A$ to be the matrix associated to the linear transformation $\partial_d\colon C_d(K)\to C_{d-1}(K)$, where we use the $d$-simplices as a basis for $C_d(K)$ and the $(d-1)$-simplices as a basis for $C_{d-1}(K)$). The weight of each column of the matrix is set to the weight of the simplex it corresponds to and the target vector be the sum of the $(d-1)$-simplices in $\BndA$. Solving this MLD problem is then precisely the same as solving the original MBC$_d$ problem. 

\begin{remark}
We have used the notation $\HasseGraph(A)$ because if $A$ is the matrix associated to a boundary map from $d$-dimensional chains, then $\HasseGraph(A)$ is the same graph as the $d$'th level of the Hasse diagram of the simplicial complex. This graph was used as a basis for one of the FPT-algorithms \cite[Theorem 5.6]{erlend_homloc} and it is has $d$ and $d-1$ simplices as vertices and face-coface pairs $(\rho,\sigma)$ as edges, see \cref{FIG:graphs_in_spaces}. It is the treewidth of this graph we talk about when we talk about the treewidth of a simplicial complex in this paper.
\end{remark}

\begin{figure}[!h]
\centering
\includegraphics[width=\textwidth]{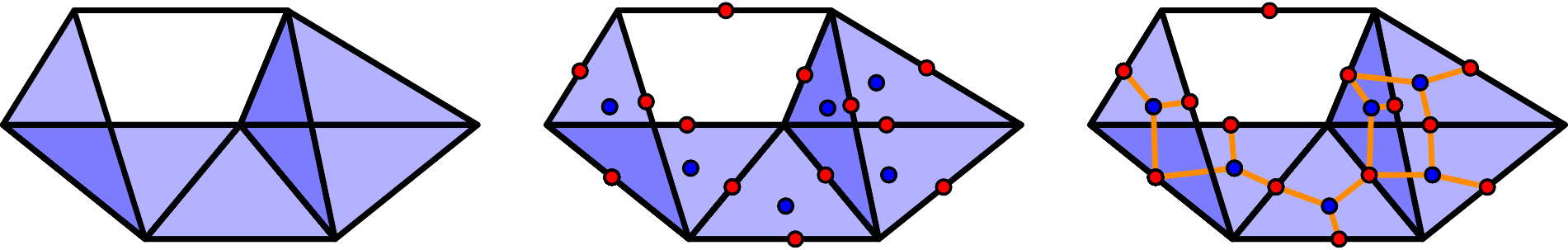}
\caption{The graph $\HasseGraph(A)$ when $A$ comes from a simplicial complex. Vertices are $(d-1)$-simplices (red) and $d$-simplices (blue), and there are edges (orange) between them if they are face-coface pairs. 
\label{FIG:graphs_in_spaces}}
\end{figure}

Our algorithm works by dynamically solving several instances of the following problem. This problem can be seen as being analogous to the \textsc{Restricted Homology Localization} (R-HL$_d$) problem (\cite[Definition 5.1]{erlend_homloc})

\begin{definition}\label{Def: Restricted HL} \textsc{Restricted GMLD} (R-GMLD): \newline 
INPUT: A bipartite graph $G$ with vertex set $(R,C)$, a set of weights $\{w_{\sigma} | {\sigma\in C}\}$, a subset $\ChnA \subseteq R$, and a four-tuple of sets $(G_t, X_t,\Qt, \Pt )$ where $X_t\subseteq G_t\subseteq R \cup C$, $\Qt \subseteq X_t \cap C$ and $\Pt \subseteq X_t \cap R$. \newline 
OUTPUT: A subset $\BndA \subseteq C$ having all the following properties:
\begin{itemize}
\itemsep0pt
    \item $(\partial \BndA\triangle \ChnA)\cap(G_t\setminus X_t) = \emptyset$.
    \item $\BndA\cap X_t= Q_t$.
    \item $(\partial \BndA\triangle \ChnA)\cap X_t = P_t$.
\end{itemize}
MINIMIZE: The sum $\sum_{\sigma\in \BndA\cap (G_t\setminus X_t)} w_\sigma$.
\end{definition}

We can think of the \textsc{Restricted GMLD} as the problem where we are free to ignore everything outside of $G_t$ and where the solutions are completely determined on $X_t\subseteq G_t$ by the sets $Q_t$ and $P_t$. This means in particular that the special case of $G_t=V(G)$ and $X_t=\emptyset$ is just the normal \textsc{GMLD} since this means that we are not ignoring anything and that no additional restrictions are placed upon the problem. 

The main idea of our algorithm for solving the \textsc{GMLD} problem is now the same as for most other treewidth based algorithms. We dynamically solve multiple instances of the restricted \textsc{GMLD} on every bag $X_t$ of a (fixed) nice tree decomposition of the bipartite graph $G$. At each bag, we store the optimal value of solutions for every pair of subsets $\Qt$ and $\Pt$. The algorithm does this by starting at the leaves working its way ``up'' towards the root, extending and combining solutions to bigger and bigger parts of the \textsc{GMLD} as we move along. At the root bag there is precisely one instance of the restricted GLMD problem to solve, namely the special case where $G_r=V(G)$ and $X_r=\emptyset$ (and so we have solved the GMLD).

\subsection{An FPT-Algorithm}\label{sec: fpt-algorithm}
Let $(T,r)$ and $X_-$ be a nice tree decomposition of the weighted bipartite graph $G$ where $V(G)= (R,C)$ and where $\ChnA\subseteq R$. For a vertex $t$ in $T$, let $G_t$ be the union of all bags $X_s$ where $s$ is a descendant of $t$. This means that we have $G_r=V(G)$ and $X_r = \emptyset$ as we promised. For every vertex $t$ in $T$ we describe how to find the weight of a minimal solution, denoted $\tab{t,Q_t,P_t}$, to the \textsc{R-GMLD} problem on $(G,U)$ restricted by the four-tuple $(G_t,X_t,Q_t,P_t)$. 

\begin{itemize}
\item \textbf{Leaf Bag}:

%Since $X_t= \emptyset$ at leaves, then $Q_t=P_t=\emptyset$ and we have 
$\tab{t,\emptyset,\emptyset}=0$.

\item \textbf{Introduce Bag}:
    
We split this into two cases. In either case the vertex $t$ in $T$ has a child $s$. First, assume that the introduced vertex $\sigma\in C$ corresponds to a column making $X_t = X_{s} \cup \{\sigma\}$. Then 
\begin{equation*}
    \tab{t,\QQt,\PPt } =
     \begin{cases}
               	\tab{s,\QQt,\PPt }          	& \sigma \not\in \QQt\\
               	\tab{s,\QQt \setminus\{\sigma\} ,\PPt \triangle (\partial \sigma\cap{X_{s}})} & \sigma\in \QQt . \\
           \end{cases}
\end{equation*}

Next assume a row vertex $\rho\in R$ is introduced so that $X_t = X_{s} \cup \{\rho\}$. If $\rho \in \PPt \triangle \partial(\QQt) \triangle \ChnA$ then there is no solution and we store the value infinity at this entry. Otherwise we have
\begin{equation*}
    \tab{t,\QQt ,\PPt } =
     \begin{cases}
               	\tab{s,\QQt  ,\PPt \setminus\{\rho\}} & \rho \in \PPt \\
               	\tab{s,\QQt ,\PPt }          	& \rho \notin \PPt .   \\
           \end{cases}
\end{equation*}

\item \textbf{Forget Bag}:

Again we have that $s$ is the child of $t$ and there are two cases. If we forget a vertex corresponding to a row $\rho\in R$ so that $X_t= X_{s} \setminus \{\rho\}$ then $\tab{t,\QQt ,\PPt } = \tab{s,\QQt ,\PPt }$. If we forget a vertex corresponding to a column $\sigma\in C$ so that $X_t = X_{s} \setminus \{\sigma\}$ then $\tab{t,\QQt ,\PPt } = \min(\tab{s,\QQt ,\PPt }, \tab{s,\QQt  \cup \{\sigma\} ,\PPt } + w_\sigma)$.

\item \textbf{Join Bag}: 

Let $s$ and $s'$ be the two children of $t$ so that $X_t=X_{s}=X_{s'}$. Then $\tab{t,\QQt ,\PPt }$ is the smallest sum $\tab{s,\QQt ,\Ps }+\tab{s',\QQt ,\Pss }$ over all pairs $P_s,P_{s'}\subseteq X_t\cap R$ such that  $\PPt  = \Ps \triangle \Pss \triangle (\partial\QQt\cap {X_t})\triangle (\ChnA\cap {X_t})$.

\end{itemize}

\begin{lemma} \label{lemma: tw algo is correct}
The above algorithm solves the GMLD problem.
\end{lemma}

\begin{proof}

Showing that this algorithm is correct requires many technical but elementary steps. We have therefore decided to omit most details and instead focus on the bigger picture. For a more detailed exposition of a similar proof, see \cite{erlend_homloc}.

We use the same basic technique for all the different bags (except for the leaf bags). The idea is to use the set of all feasible solutions  $\mathcal{S}(t,Q_t,P_t)$ (i.e. all solutions, both optimal and non-optimal) to instances of the restricted \textsc{GMLD} problems to argue that each of the formulas in the algorithm are correct. We do this by first showing that the left hand side of the equation is smaller than or the same size as the right hand side, and then to show the opposite. 

\begin{itemize}
\item \textbf{Leaf Bag}: We have that $G_t = \emptyset$ and so there is only one problem instance to solve: $G_t = X_t = Q_t = P_t = \emptyset$. There is only one solution to this problem, which is $\BndA = \emptyset$ and this solution has weight $0$.
\item \textbf{Introduce Bag}: When a column $\sigma$ is introduced there are two cases as it is either in $Q_t$ or it is not. In the first case we can show that $\BndA \in \mathcal{S}(t,Q_t,P_t)$ if and only if $\BndA \in \mathcal{S}(s,Q_t,P_t)$, and in the latter we have to show that $\BndA \in \mathcal{S}(t,Q_t,P_t)$ if and only if $\BndA \setminus \{\sigma\} \in \mathcal{S}(s,Q_t\setminus \{\sigma\},P_t\triangle(\partial \sigma \cap X_s))$. The details of this are elementary, using the fact that an introduced node in a nice tree decomposition is not adjacent to any forgotten nodes. The case where a row vertex is introduced is analogous.

\item \textbf{Forget Bag}: When a column $\sigma$ is forgotten we show that $\BndA \in \mathcal{S}(t,Q_t,P_t)$ if and only if $\BndA \in \mathcal{S}(s,Q_t,P_t)\cup \mathcal{S}(s,Q_t\cup \{\sigma\},P_t)$. If $\BndA \in \mathcal{S}(t,Q_t\cup \{\sigma\},P_t)$ then $\BndA \in \mathcal{S}(s,Q_t,P_t)$ also has to take the weight of $\sigma$. When a row $\rho$ is forgotten we show that $\BndA \in \mathcal{S}(t,Q_t,P_t)$ if and only if $\BndA \in \mathcal{S}(t,Q_t,P_t)$. The reason why we don't also have to think about $\BndA \in \mathcal{S}(t,Q_t,P_t \cup \{\rho\})$ is that this is no longer a solution as the row $\rho$ would be in the boundary of $\BndA$. 

\item \textbf{Join Bag}: In the last case we have a solution $\BndA \in \mathcal{S}(t,Q_t,P_t)$ if and only if there is a pair of solutions $\BndB \in \mathcal{S}(s,Q_t,P_s)$ and $\BndC \in \mathcal{S}(s',Q_t,P_{s'})$ such that $P_t = P_s\triangle P_{s'}\triangle (\ChnA \cap X_t)\triangle (\partial \BndA \cap X_t)$. The idea here is to show that if we set $\BndA = \BndB\cup \BndC$ then the relation between $P_t,\, P_s$ and $P_{s'}$ is exactly the one described above. For the other way we let $\BndB=\BndA\cap G_s$ and $\BndC=\BndA\cap G_{s'}$.

\end{itemize}
\end{proof}

The above treewidth algorithm terminates in $\mathcal{O}(4^{\tau}n)$ time, where $\tau$ is the treewidth of $ \HasseGraph(A)$. To see this, note that at each introduce and forget bag the algorithm has to compute at most $2^\tau$ values, each taking constant time. Meanwhile at the join bag the algorithm computes at most $2^\tau$ values where each is the minimum of $2^\tau$ numbers, which means that the join bag takes $\mathcal{O}(4^{\tau}n)$ time. The number of bags is linear in input size, so we get our result. Together with \cref{lemma: tw algo is correct}, this discussion proves \cref{thm:treewidth-algorithm}.

The algorithm can be made to return an optimal solution in $\mathcal{O}(4^{\tau}n)$ time by backtracking through the tables of solutions for each bag. Using the more naive approach of keeping track of a representative optimal solution is also possible. This would give us a worse runtime of $\mathcal{O}(4^{\tau}n^2)$ as we need to copy and store partial solutions which may have size linear in $n$. The algorithm can also be used to find a maximum cycle, since it works even when the weights are negative.

\subsection{ETH-tightness}\label{sec:eth-tightness}
In this final subsection we show that the treewidth based FPT algorithm we just discussed is ETH-tight, by proving \cref{thm:treewidth-hardness}. In fact, this theorem is true even for the SMBC$_d$ problem. 

\begin{thm}
The unweighted SMBC$_d$ problem can not be solved in $2^{o(\tau)}\operatorname{poly}(n)$-time, assuming the ETH. 
\end{thm}

Taking suspension doubles the treewidth of a space \cite[Sec. 2.2]{erlend_homloc}, so it is sufficient to show the result for $d=2$, and the general case follows by inductively taking the suspension. The reduction we use to prove this result builds on the reduction from \textsc{Max Cut} parameterized by treewidth presented in \cite[Sec. 6]{erlend_homloc}. This reduction can in turn be thought of as a specialized version of the reduction presented in \cite{borradaile_et_al:LIPIcs:2020:12179}. Recall that a cut in a graph is just a partitioning of the vertices of a graph $G$ into two sets $I$ and $J$, and the size of the cut is the number of edges crossing the cut (see \cref{FIG: max cut illustration}).

\begin{figure}[!h]
\centering
  \includegraphics[width=0.7\textwidth]{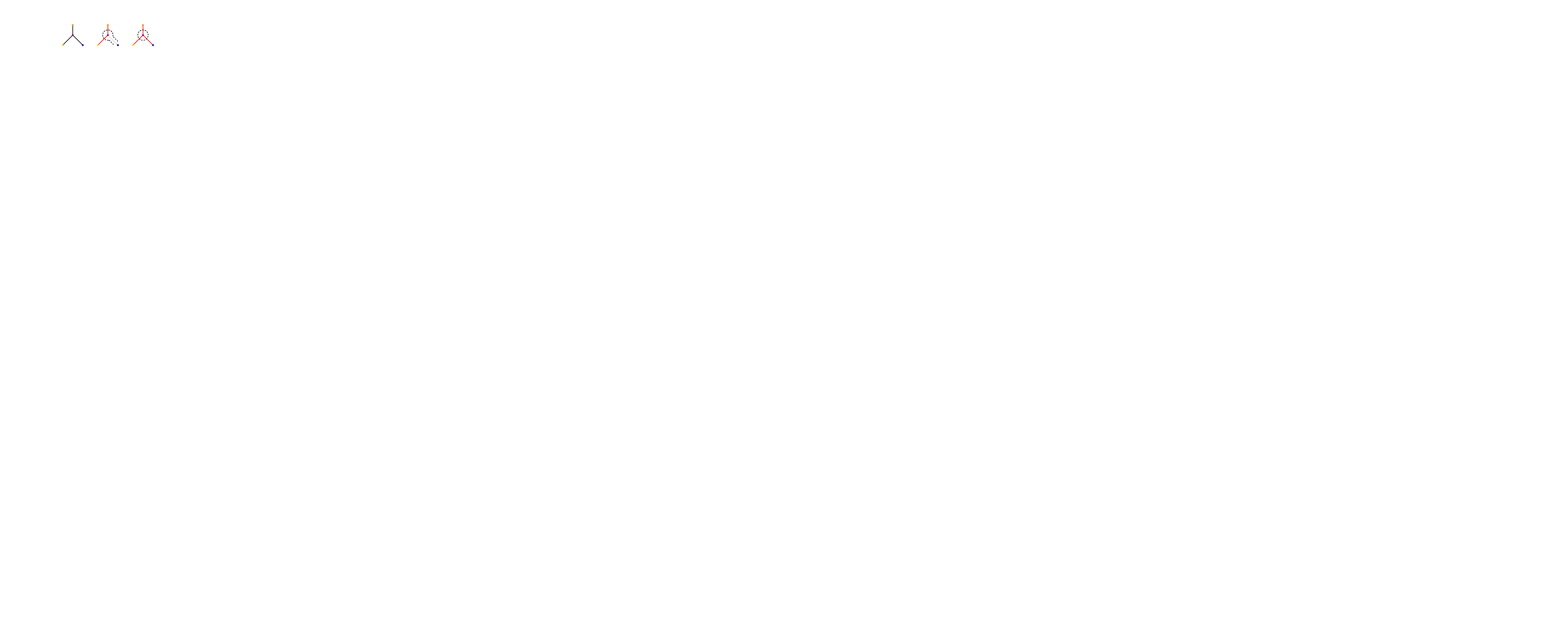}
 \caption{\label{FIG: max cut illustration} A graph $G$ (left) with two different cuts indicated by the dashed line in (middle and right). The edges crossing the cuts are marked in red. }
\end{figure}

\begin{definition} The \textsc{Max Cut} problem\\
INPUT: A graph $G$ on $n$ vertices.\\
OUTPUT: A cut $(I,J)$ in $G$.\\
MAXIMIZING: The size of the cut $(I,J)$.
\end{definition}

Our reduction maps a graph $G$ to the space $Y(G)$. We think of $Y(G)$ as the quotient of three sub-spaces (see the example in \cref{FIG: three componets}):

\begin{enumerate}
    \itemsep0pt
    \item $Y_P(G)$: A pair of pants with $|E(G)|$ legs, one for each edge in $G$. The $1$-simplices contained in the ``waist'' $x$ of this pair of pants is the input boundary $\ChnA$ to the MBC$_2$ problem.
    \item $Y_S(G)$: The $2$-dimensional (orientable manifold) simplicial complex obtained by associating a $2$-sphere to every vertex $v$ of $G$ and take the connected sum of neighboring vertices. Two spheres intersects in a circle if there is an edge between the corresponding vertices, and the leg of $Y_P(G)$ corresponding this edge is glued to this circle. This subspace looks like the surface of some thickening of the graph $G$.
    \item $Y_D(G)$: Finally, glue a disk along its boundary to each such intersection-circles.
\end{enumerate}

\begin{figure}[!h]
\centering
  \includegraphics[width=0.7\textwidth]{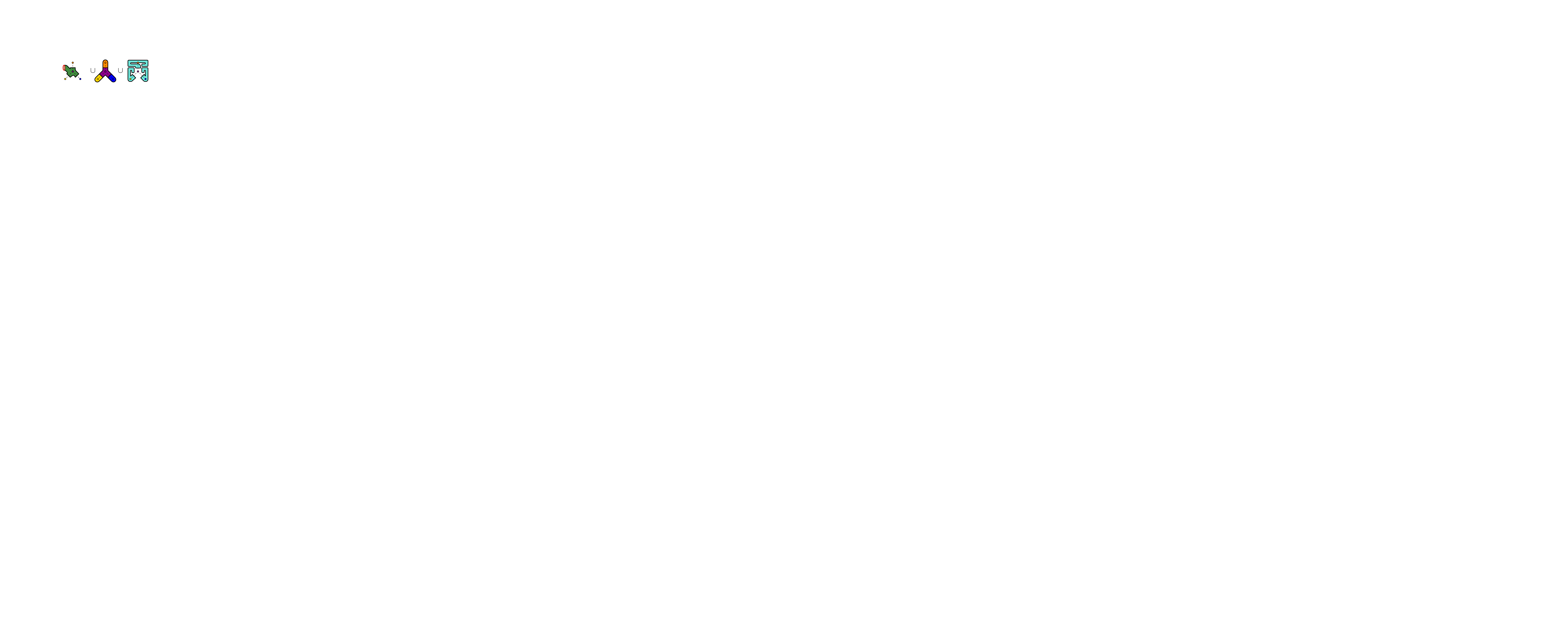}
 \caption{\label{FIG: three componets} A figure showing the three components of the space $Y(G)$ reduced from the graph $G$ in \cref{FIG: max cut illustration} (the dashed lines). The pair of pants $Y_P(G)$ (left), the connected spheres (or ``thick graph'') $Y_S(G)$ (middle) and the disks $Y_D(G)$ (right). The input chain of the MBC$_2$ problem is colored in red.}
\end{figure}

We can always find a bounding chain in $Y(G)$, by taking the pair of pants and all disks in $Y_D(G)$. The idea is that the disks we have in the bounding chain correspond to edges that are not cut. So to maximize the size of the cut, we want to minimize the number of disks in our solution. To achieve this, we make the disks as big (consisting of many simplices) as possible, by subdividing them sufficiently many times.

Mapping solutions back and forth is quite intuitive. % (see \cref{FIG: correspondence between solutions}). 
Given a cut $(I,J)$, the minimum bounding chain of the $1$-simplices in the waist $x$ consists of the pair of pants, each of the spheres corresponding to vertices in $I$ and each of the disks corresponding to edges that are not cut. Conversely, starting with a bounding chain $\BndA$ of the waist, we get a cut $(I,J)$ by letting a vertex $v$ be in $I$ if a $2$-simplex (and hence every $2$-simplex) of the sphere corresponding to $v$ is in $\BndA$, otherwise it is in $J$.

We are left with the task of finding a triangulation of the space $Y(G)$ of low treewidth. We describe how the techniques developed in \cite{erlend_homloc} can be altered to work for the SMBC$_d$ problem. 

The fundamental idea is to let the triangulation of the space depend on some (arbitrary) nice tree decomposition of the input graph $G$ of low treewidth. So the first step of the reduction would be to compute such a nice tree decomposition $\text{TD}(G)$ (to within a constant factor approximation of the actually treewidth), which we know can be done in $2^{\mathcal{O}(\tau)}n$-time. A concrete example of what a nice tree decomposition and the corresponding space typically look like is pictured in \cref{FIG: Nice TD example}.

\begin{figure}[!h]
\centering
\begin{subfigure}{0.47\textwidth}
\centering
  \includegraphics[width=0.95\textwidth]{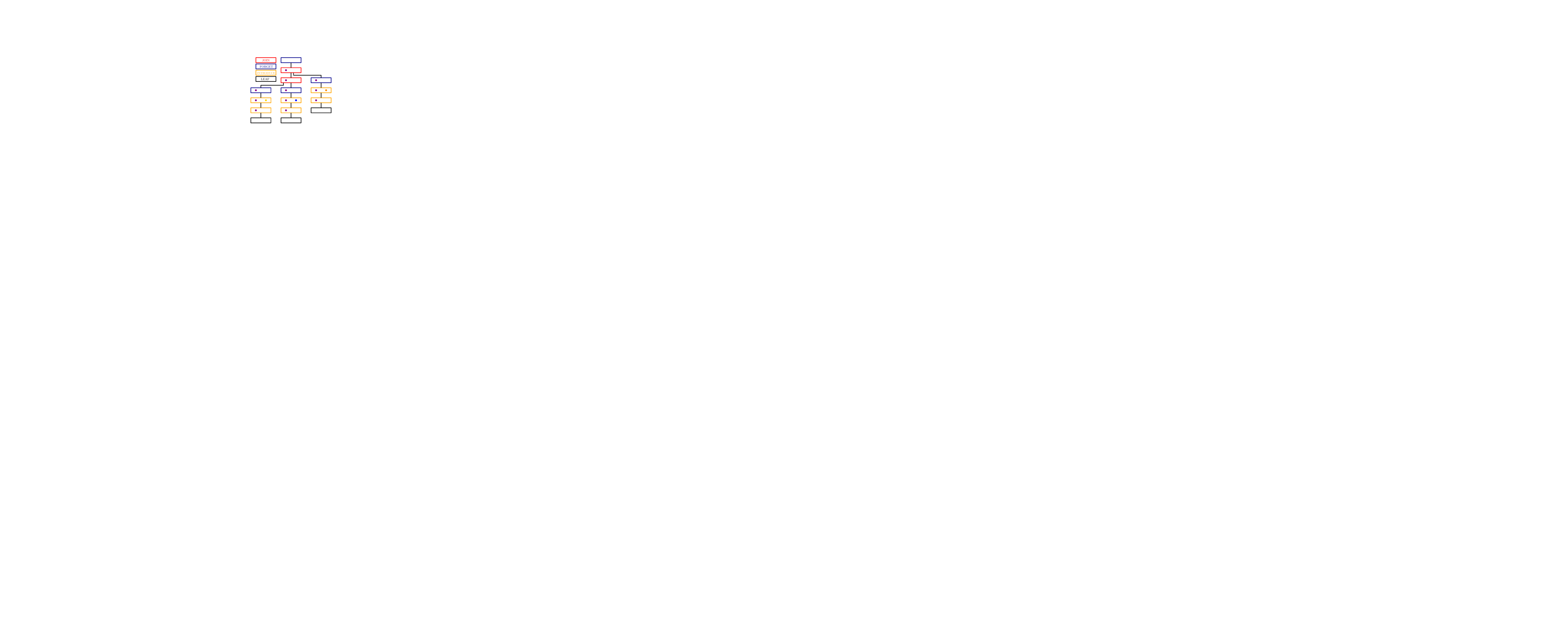}
\end{subfigure}
\begin{subfigure}{0.51\textwidth}
\centering
  \includegraphics[width=0.95\textwidth]{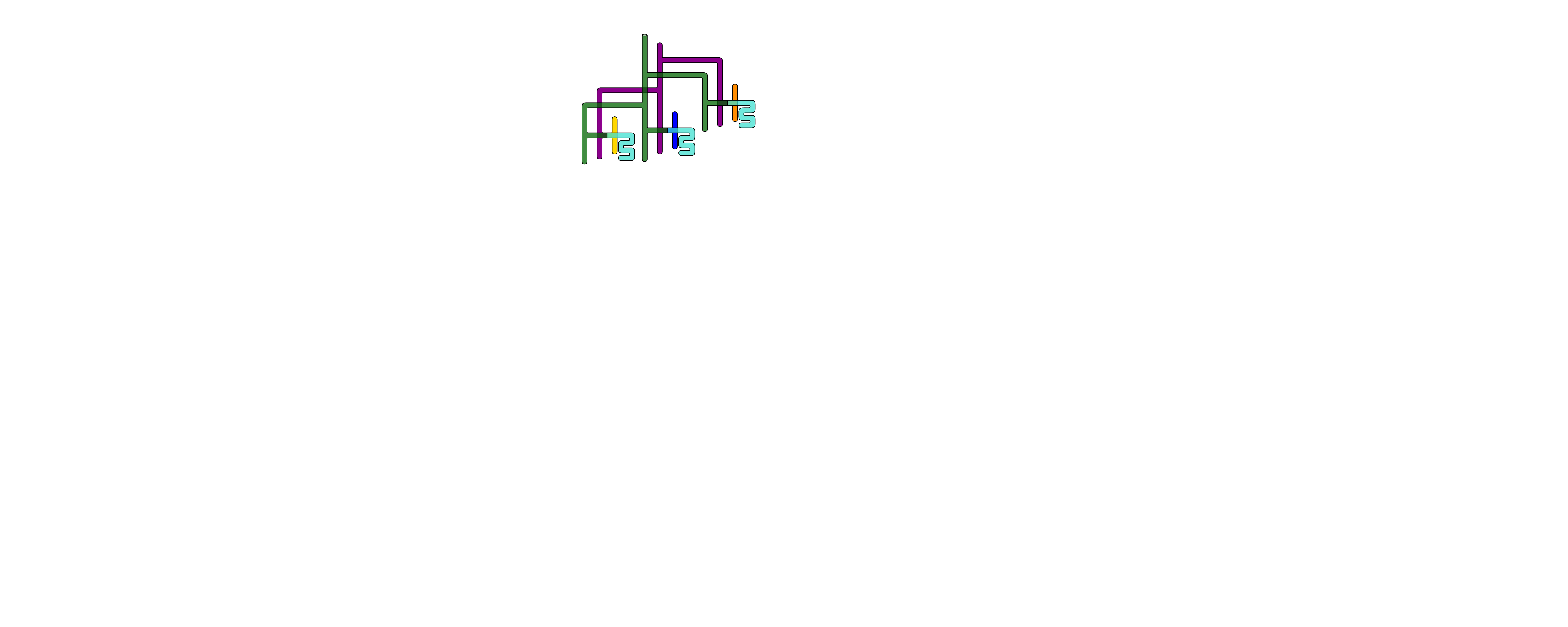}
\end{subfigure}
\caption{\label{FIG: Nice TD example}A nice tree decomposition (left) of the graph $G$ from \cref{FIG: max cut illustration} where the image of the function $X_-$ is pictured inside each vertex, and how to shape/triangulate the space $Y(G)$ to keep its treewidth down.}
\end{figure}

The next step is to  give the idea of how to triangulate the space. The subspace $Y_S(G)$ is triangulated like \cite[Sec. 6,2]{erlend_homloc}. Intuitively, we deform the spheres representing each vertex $v$ into long ``tubes'' following the shape of the bags containing $v$ in the nice tree decomposition of the input graph. The underlying graph of such a space is of low treewidth, as we can find a tree decomposition of treewidth linear in the treewidth of $G$ \cite[Lemma 6.5]{erlend_homloc}.

The pair of pants $Y_P(G)$ is triangulated in a similar way, by thinking of it as a sphere corresponding to a vertex that is in every non-empty bag of the nice tree decomposition. Finally, the disks in $Y_D(G)$ is stretched out like long cylinders/tubes, keeping in mind that we want them containing many simplices, but covered by small bags of the tree decomposition. The final result looks like in \cref{FIG: Nice TD example}.

\section{Conclusion}\label{sec:conclusions}

We have shown that the MBC$_d$ problem is difficult to solve even for spaces with small coface degree or small solution sizes and provided two parameterized algorithms for solving the MBC$_d$ problem. However, there are still many unanswered questions and interesting research directions to explore.

\subsection{The \texorpdfstring{MBC$_1$}{MBC\_1} problem}

The polynomial time algorithm for the MBC$_1$ problem presented in this paper was chosen because it was easy to describe. Recall that we essentially solved the problem by combining two algorithms. First we use the Floyd-Warshall algorithm to construct a distance matrix from the $1$-skeleton of the input space. Then we use any polynomial time algorithm solving the minimal weighted matching problem on a sub-matrix of the distance matrix (viewed as a complete graph). It would be interesting to know if this problem can be solved more efficiently using a more carefully designed algorithm.

\subsection{The Dijkstra approach}

The techniques we developed in \cref{sec:dijkstra} open up many new directions of further research, both theoretical and practical, that we think deserves some attention. Recall that $k$ denote solution size and that $c$ denote coface degree while $n$ is the number of $d-1$ simplices.
\begin{itemize}
\itemsep0pt
    \item Can we bridge the gap between the $2^{\mathcal{O}(k\log(c))}\operatorname{poly}(n)$-time algorithm and the $2^{o(\sqrt{k}\log(c))}\operatorname{poly}(n)$-time ETH lower bound presented in this paper? 
    \item Is there a constant $0<e<1$ for which we can solve the MBC$_d$ problem in $c^{ek}\operatorname{poly}(n)$ time? We suspect that it is possible to solve the MBC$_2$ problem in $\mathcal{O}(c^{\frac{1}{3}k}n)$-time, by only looking at simplicies if they are at most $k/3$ ``simplices away'' from the boundary (instead of $k$).
    \item Can we modify our algorithm so that it can be used to find the smallest $2$-manifold in a simplicial complex with a particular boundary?
    \item Is there a kernel smaller than the obvious one\footnote{I.e. the kernel consisting of every $(d+1)$-simplex that is at most $k$ simplices away.} of size $c^k$ for the MBC$_d$ problem? 
    \item The Dijkstra based algorithm is well suited as a basis for using various kinds of A*-type heuristics. It would therefore be very interesting to explore how this can be used to speed up computations in practice.
\end{itemize}

\subsection{Treewidth}

There are also a open questions surrounding our results on treewidth of the $d$'th level of the Hasse diagram. This is perhaps particularly interesting as there seems to be an increasing interest in the use of treewidth techniques in computational topology.
\begin{itemize}
\itemsep0pt
    \item Is it possible to solve the MBC$_d$ problem in $2^\tau \operatorname{poly}(n)$-time?
    \item Can we implement better treewidth algorithms (e.g. by using massive parallelization) that are competitive with ILP-solvers?
    \item We can use a slight modification of our ETH-reduction to prove that finding a $2$-manifold with a a particular boundary (or genus) in a simplicial complex cannot be done in $2^{o(\tau)} \operatorname{poly}(n)$-time (assuming the ETH). Black and Nayyeri proved  that this problem can be solved in $\tau_1^{\mathcal{O}(\tau_1 ^2)}\operatorname{poly}(n)$-time in \cite[Theorem 1.1]{black2021finding}, where $\tau_1$ is the treewidth of the $1$-skeleton of the simplicial complex. This leaves a gap down to our lower bound which it would be interesting to see if could be bridged.
\end{itemize}

\subsection{Applications in Topological Data Analysis}
We began working on this project because we were interested in designing algorithms for finding geometrically concise representatives for cycles in persistent homology. Though it turned into a paper on theoretical computer science in the end, we have described a polynomial time algorithm for finding the shortest $1$-cycle born at a given filtration value as well as two FPT-algorithm solving the same problem in higher dimensions. We look forward to exploring how information about these representatives may be included in the persistent homology toolkit in the future.

\section*{Acknowledgement}
We wish to thank Daniel Lokshtanov for his help in proving \cref{LEMMA: Daniel L}. Erlend Raa V{\aa}gset acknowledges support from the Research Council of Norway grant “Parameterized Complexity for Practical Computing (PCPC)” (NFR, no. 274526).

\bibliography{bibliography}

\begin{thebibliography}{10}
\expandafter\ifx\csname url\endcsname\relax
  \def\url#1{\texttt{#1}}\fi
\expandafter\ifx\csname urlprefix\endcsname\relax\def\urlprefix{URL }\fi
\expandafter\ifx\csname href\endcsname\relax
  \def\href#1#2{#2} \def\path#1{#1}\fi

\bibitem{Grady2010}
L.~Grady, Minimal surfaces extend shortest path segmentation methods to 3d,
  IEEE Transactions on Pattern Analysis and Machine Intelligence 32~(2) (2010)
  321--334.
\newblock \href {https://doi.org/10.1109/TPAMI.2008.289}
  {\path{doi:10.1109/TPAMI.2008.289}}.

\bibitem{escolar}
E.~G. Escolar, Y.~Hiraoka, Optimal cycles for persistent homology via linear
  programming, in: Optimization in the Real World, Springer Japan, Tokyo, 2016,
  pp. 79--96.

\bibitem{emmett2015multiscale}
K.~Emmett, B.~Schweinhart, R.~Rabadan, Multiscale topology of chromatin folding
  (2015).
\newblock \href {http://arxiv.org/abs/1511.01426} {\path{arXiv:1511.01426}}.

\bibitem{Sullivan1990}
J.~Sullivan, A crystalline approximation theorem for hypersurfaces, Ph.D.
  thesis, Princeton University (1990).

\bibitem{Dunfield2011}
N.~M. Dunfield, A.~N. Hirani,
  \href{https://doi.org/10.1145/1998196.1998218}{The least spanning area of a
  knot and the optimal bounding chain problem}, in: Proceedings of the
  Twenty-Seventh Annual Symposium on Computational Geometry, SoCG '11,
  Association for Computing Machinery, New York, NY, USA, 2011, p. 135–144.
\newblock \href {https://doi.org/10.1145/1998196.1998218}
  {\path{doi:10.1145/1998196.1998218}}.
\newline\urlprefix\url{https://doi.org/10.1145/1998196.1998218}

\bibitem{Chambers2015}
E.~W. Chambers, M.~Vejdemo-Johansson,
  \href{https://onlinelibrary.wiley.com/doi/abs/10.1111/cgf.12514}{Computing
  minimum area homologies}, Computer Graphics Forum 34~(6) (2015) 13--21.
\newblock \href {https://doi.org/https://doi.org/10.1111/cgf.12514}
  {\path{doi:https://doi.org/10.1111/cgf.12514}}.
\newline\urlprefix\url{https://onlinelibrary.wiley.com/doi/abs/10.1111/cgf.12514}

\bibitem{borradaile_et_al:LIPIcs:2020:12179}
G.~Borradaile, W.~Maxwell, A.~Nayyeri,
  \href{https://drops.dagstuhl.de/opus/volltexte/2020/12179}{{Minimum Bounded
  Chains and Minimum Homologous Chains in Embedded Simplicial Complexes}}, in:
  S.~Cabello, D.~Z. Chen (Eds.), 36th International Symposium on Computational
  Geometry (SoCG 2020), Vol. 164 of Leibniz International Proceedings in
  Informatics (LIPIcs), Schloss Dagstuhl--Leibniz-Zentrum f{\"u}r Informatik,
  Dagstuhl, Germany, 2020, pp. 21:1--21:15.
\newblock \href {https://doi.org/10.4230/LIPIcs.SoCG.2020.21}
  {\path{doi:10.4230/LIPIcs.SoCG.2020.21}}.
\newline\urlprefix\url{https://drops.dagstuhl.de/opus/volltexte/2020/12179}

\bibitem{chen2011hardness}
C.~Chen, D.~Freedman, Hardness results for homology localization, Discrete \&
  Computational Geometry 45~(3) (2011) 425--448.

\bibitem{erlend_homloc}
N.~Blaser, E.~R. Vågset, Homology localization through the looking-glass of
  parameterized complexity theory (2020).
\newblock \href {http://arxiv.org/abs/2011.14490} {\path{arXiv:2011.14490}}.

\bibitem{mld}
E.~R. Berlekamp, R.~J. McEliece, H.~C. van Tilborg, On the inherent
  intractability of certain coding problems, IEEE Transactions on Information
  Theory 24~(3) (1978) 385--386.

\bibitem{black2021finding}
M.~Black, A.~Nayyeri, Finding surfaces in simplicial complexes with
  bounded-treewidth 1-skeleton (2021).
\newblock \href {http://arxiv.org/abs/2107.10339} {\path{arXiv:2107.10339}}.

\bibitem{DowneyFellows2013}
R.~G. Downey, M.~R. Fellows, Fundamentals of Parameterized Complexity, Springer
  Publishing Company, Incorporated, 2013.

\bibitem{PA}
M.~Cygan, F.~V. Fomin, {\L}.~Kowalik, D.~Lokshtanov, D.~Marx, M.~Pilipczuk,
  M.~Pilipczuk, S.~Saurabh, Parameterized Algorithms, Springer International
  Publishing, 2015.
\newblock \href {https://doi.org/10.1007/978-3-319-21275-3}
  {\path{doi:10.1007/978-3-319-21275-3}}.

\bibitem{floyd}
R.~W. Floyd, \href{https://doi.org/10.1145/367766.368168}{Algorithm 97:
  Shortest path}, Commun. ACM 5~(6) (1962) 345.
\newblock \href {https://doi.org/10.1145/367766.368168}
  {\path{doi:10.1145/367766.368168}}.
\newline\urlprefix\url{https://doi.org/10.1145/367766.368168}

\bibitem{edmonds_1965}
J.~Edmonds, Paths, trees, and flowers, Canadian Journal of Mathematics 17
  (1965) 449–467.
\newblock \href {https://doi.org/10.4153/CJM-1965-045-4}
  {\path{doi:10.4153/CJM-1965-045-4}}.

\bibitem{fredman}
M.~Fredman, R.~Tarjan, Fibonacci heaps and their uses in improved network
  optimization algorithms, J. ACM 34 (1987) 596--615.

\bibitem{Burton2019}
B.~Burton, S.~Cabello, S.~Kratsch, W.~Pettersson, The parameterized complexity
  of finding a 2-sphere in a simplicial complex, SIAM Journal on Discrete
  Mathematics 33~(4) (2019) 2092--2110.
\newblock \href {https://doi.org/10.1137/18M1168704}
  {\path{doi:10.1137/18M1168704}}.

\bibitem{fellows2007fixed}
M.~R. Fellows, D.~Hermelin, F.~Rosamond, On the fixed-parameter intractability
  and tractability of multiple-interval graph problems, Unpublished Result
  (2007).

\bibitem{bodlaender2016c}
H.~L. Bodlaender, P.~G. Drange, M.~S. Dregi, F.~V. Fomin, D.~Lokshtanov,
  M.~Pilipczuk, A c\^{}kn 5-approximation algorithm for treewidth, SIAM Journal
  on Computing 45~(2) (2016) 317--378.

\end{thebibliography}
%\printbibliography 

\end{document}